\documentclass[12pt,a4paper,twoside]{article}

\usepackage{fancyhdr,graphicx,subfigure,psfrag}
\usepackage{amsmath,amsfonts,amsthm}
\usepackage[abs]{overpic}
\usepackage{epsfig}

 \DeclareMathOperator{\diag}{diag}

\newcommand{\Lie}{L}
\newcommand{\Res}{\mathrm{Res}}

\newcommand{\D}{\mathrm{d}}

\newcommand{\tr}{\mathrm{tr}}

\newtheorem{theorem}{Theorem}
\newtheorem{lemma}{Lemma}
\newtheorem{corollary}{Corollary}
\newtheorem{remark}{Remark}
\newtheorem{proposition}{Proposition}
\newtheorem{definition}{Definition}

\numberwithin{equation}{section}
\begin{document}
\title{Universality in the two matrix model with a
monomial quartic and a general even polynomial potential}
\author{M. Y. Mo}
\date{}
\maketitle
\begin{abstract}
In this paper we studied the asymptotic eigenvalue statistics of the
2 matrix model with the probability measure
\begin{equation*}
Z^{-1}_{n}\exp\left(-n\left(\tr(V(M_1)+W(M_2)-\tau
M_1M_2\right)\right)\D M_1\D M_2,
\end{equation*}
in the case where $W=\frac{y^4}{4}$ and $V$ is a general even
polynomial. We studied the correlation kernel for the eigenvalues of
the matrix $M_1$ in the limit as $n\rightarrow\infty$. We extended
the results of Duits and Kuijlaars in \cite{DK} to the case when the
limiting eigenvalue density for $M_1$ is supported on
multiple intervals. The results are achieved by constructing the
parametrix to a Riemann-Hilbert problem obtained in \cite{DK} with
theta functions and then showing that this parametrix is
well-defined for all $n$ by studying the theta divisor.
\end{abstract}
\section{Introduction}
\subsection{2 matrix models and biorthogonal polynomials}
The 2-matrix Hermitian models are matrix models with the probability
measure
\begin{equation}\label{eq:2MM}
Z^{-1}_{n}\exp\left(-n\left(\tr(V(M_1)+W(M_2)-\tau
M_1M_2\right)\right)\D M_1\D M_2,
\end{equation}
defined on the space of pairs $(M_1,M_2)$ of $n\times n$ Hermitian
matrix. The constant $Z_{n}$ is the normalization constant of the
measure, $\tau\in\mathbb{R}\setminus\{0\}$ is a coupling constant
and $\D M_1$, $\D M_2$ are the standard Lebesgue measures on the
space of Hermitian matrices. In (\ref{eq:2MM}), $V$ and $W$ are
called potentials of the matrix model. In this paper, we shall
consider $V$ to be a general even polynomial and $W$ to be
the monomial $W(y)=\frac{y^4}{4}$.

Let $x_1,\ldots,x_n$ and $y_1,\ldots,y_n$ be the eigenvalues of the
matrices $M_1$ and $M_2$ respectively, then the eigenvalues of the
matrix model (\ref{eq:2MM}) are distributed according to
\begin{equation}\label{eq:jpdf}
\mathcal{P}(\vec{x},\vec{y})=\tilde{Z}^{-1}_{n}\prod_{i<k}^{n}(x_i-x_k)^2(y_i-y_k)^2e^{-n\left(\sum_{j=1}^nV(x_j)+W(y_j)-\tau
x_jy_j\right)}
\end{equation}
where $\tilde{Z}_{n}$ is a normalization constant and
$\vec{x}=(x_1,\ldots,x_n)$, $\vec{y}=(y_1,\ldots,y_n)$.

The two-matrix model was introduced in \cite{IZ}, \cite{M2} as a
generalization of the one matrix model to study critical phenomena
in physical systems. The 2 matrix model is needed to represent all
conformal models in statistical physics \cite{DKK}. It is also a
powerful tool in the studies of random surfaces as the large $N$
expansion of the partition function $\tilde{Z}_{n}$ is expected to
be the generating function of discretized surface \cite{Ka}. Since
its introduction, the 2 matrix model has become a very active
research area \cite{Advan}, \cite{BerEy}, \cite{Bsemi}, \cite{BeEy},
\cite{BEHmulti}, \cite{BEHRH}, \cite{Du}, \cite{DK}, \cite{EMc},
\cite{Ey1}, \cite{Ey2MM}, \cite{EylargeN}, \cite{EyMe},
\cite{EyOr1}, \cite{EyOr2}, \cite{Kap}, \cite{KMRH} and one of the
major problems is to obtain rigorous asymptotics for the eigenvalue
statistics. A good review of the subject can be found in
\cite{Dfr1}, \cite{Dfr2}.

A particular important object in the studies of eigenvalue
statistics is the correlation function
$\mathcal{R}_{m,l}^{n}(x_1,\ldots,x_m,y_1,\ldots,y_l)$
\begin{equation}\label{eq:corel}
\mathcal{R}_{m,l}^{n}(x_1,\ldots,x_m,y_1,\ldots,y_l)=\frac{(n!)^2}{(n-m)!(n-l)!}
\int_{\mathbb{R}^{n-m}}\int_{\mathbb{R}^{n-l}}\mathcal{P}(\vec{x},\vec{y})\prod_{j=m+1}^n\D
x_{j}\prod_{k=l+1}^{n}\D y_{k}.
\end{equation}

In \cite{Ey1}, \cite{MS}, a connection between biorthogonal
polynomials and the correlation functions of 2 matrix models
(\ref{eq:corel}) was found. Let $p_{k}(x)$ and $q_l(y)$ be monic
polynomials of degrees $k$ and $l$ such that
\begin{equation}\label{eq:bop}
\int_{\mathbb{R}}\int_{\mathbb{R}}\D x\D y
p_k(x)q_l(y)e^{-n\left(V(x)+W(y)-\tau xy\right)}=h_{k}\delta_{kl},
\end{equation}
for some constants $h_k$, then these polynomials exist and are
unique \cite{BEHmulti}, \cite{EMc}. These polynomials are known as
biorthogonal polynomials.

Let us define some integral transforms of the biorthogonal
polynomials by
\begin{equation}\label{eq:trans}
\begin{split}
Q_k(x)&=e^{-nV(x)}\int_{\mathbb{R}}q_k(y)e^{-n\left(W(y)-\tau
xy\right)}\D y,\\
P_k(y)&=e^{-nW(y)}\int_{\mathbb{R}}p_k(x)e^{-n\left(V(x)-\tau
xy\right)}\D x,
\end{split}
\end{equation}
and define the kernels to be
\begin{equation}\label{eq:ker}
\begin{split}
K_{11}^{n}(x_1,x_2)&=\sum_{j=0}^{n-1}\frac{1}{h_j}p_j(x_1)Q_j(x_2),\quad
K_{22}^{n}(x,y)=\sum_{j=0}^{n-1}\frac{1}{h_j}P_j(x)q_j(y),\\
K_{12}^{n}(y,x)&=\sum_{j=0}^{n-1}\frac{1}{h_j}p_j(y)q_j(x),\\
K_{21}^{n}(y_1,y_2)&=\sum_{j=0}^{n-1}\frac{1}{h_j}P_j(y_1)Q_j(y_2)-e^{-n\left(V(x)+W(y)-\tau
xy\right)}.
\end{split}
\end{equation}
Then the correlation function $\mathcal{R}_{m,l}^{n}$
(\ref{eq:corel}) has the following determinantal expression.
\begin{equation}\label{eq:det}
\mathcal{R}_{m,l}^{n}(x_1,\ldots,x_m,y_1,\ldots,y_l)=\det\begin{pmatrix}\left(K_{11}^{n}(x_i,x_j)\right)_{i,j=1}^{m}&
\left(K_{12}^{n}(x_i,y_j)\right)_{i,j=1}^{m,l}\\
\left(K_{21}^{n}(y_i,x_j)\right)_{i,j=1}^{l,m}&\left(K_{22}^{n}(y_i,y_j)\right)_{i,j=1}^{l}\end{pmatrix}.
\end{equation}
Upon averaging over the variables $y_k$, we see that the $m$-point
correlation function $\mathcal{R}_{m,0}^{n}(x_1,\ldots,x_m)$ for the
eigenvalues of the matrix $M_1$ is given by the kernel
$K_{11}^{n}(x_i,x_j)$,
\begin{equation}\label{eq:ker1}
\mathcal{R}_{m,0}^{n}(x_1,\ldots,x_m)=\det\left(K_{11}^{n}(x_i,x_j)\right)_{i,j=1}^{m}
\end{equation}
The purpose of this paper is to provide a rigorous asymptotic
expression for the kernel $K_{11}^{n}$ as $n\rightarrow\infty$ for
$W(y)=\frac{y^4}{4}$ and $V(x)$ being a general even polynomial.

Due to a generalized Christoffel-Darboux formula, the kernels
$K_{11}^{n}$ can be expressed in terms of a finite sum of the
biorthogonal polynomials (See \cite{BEHmulti}, \cite{EylargeN},
\cite{UN} and \cite{DaK}). This reduces the problem of finding an
asymptotic expression for $K_{11}^{n}$ into finding the asymptotics
of the biorthogonal polynomials.

\subsection{Rigorous results in the ``one-cut regular" case}

Until recently, most results in the asymptotics of biorthogonal
polynomials have been obtained through heuristic argument (See
\cite{EylargeN}, \cite{Ey2MM}). For a long time, the only rigorous
result was the case when both $W(y)$ and $V(x)$ are quadratic
polynomials \cite{EMc}. In the recent work by Duits and Kuijlaars
\cite{DK}, (See also \cite{Du}, Chapter 5, which was later made into
the publication \cite{DK}), the Deift-Zhou steepest decent method
(\cite{D}, \cite{DKV}, \cite{DKV2}, see also \cite{BI}) was
successfully applied to obtain the asymptotics of biorthogonal
polynomials with $W(y)=\frac{y^4}{4}$ and $V(x)$ an even
polynomial in the case when the limiting eigenvalue density for
$M_1$ is supported on a single interval. The main idea in \cite{DK}
is to transform and approximate the Riemann-Hilbert problem
satisfied by the biorthogonal polynomials \cite{KMRH}, \cite{BEHRH}
(See Section \ref{se:RHP} for details) via the use of suitable
equilibrium measures and then solve the approximated Riemann-Hilbert
problem explicitly to obtain asymptotic formula for the biorthogonal
polynomials. The results in \cite{DK} was obtained in the case when
one of these equilibrium measures is supported on a single interval.
This case is called the ``one-cut regular case" in \cite{DK}.

To be precise, let $I(\nu_i,\nu_j)$ be the following energy
function
\begin{equation}\label{eq:logen}
I(\nu_i,\nu_j)=\int\int\log\left(\frac{1}{|x-y|}\right)\D\nu_i(x)\D\nu_j(y),
\end{equation}
where the integral is performed on the supports of the measures
$\nu_i$ and $\nu_j$. Then the equilibrium measures $\mu_1$, $\mu_2$
and $\mu_3$ are the measures that minimize the following energy
function $E_V(\nu_1,\nu_2,\nu_3)$.
\begin{definition}\label{de:equ}
(Definition 2.2 in \cite{DK} (See also Definition 5.2.1 in
\cite{Du})) The equilibrium measure $(\mu_1,\mu_2,\mu_3)$ is the
triplet of measures that minimizes the following energy function.
\begin{equation}\label{eq:EV}
\begin{split}
E_V(\nu_1,\nu_2,\nu_3)=\sum_{j=1}^3I(\nu_j,\nu_j)-\sum_{j=1}^2I(\nu_j,\nu_{j+1})
+\int\left(V(x)-\frac{3}{4}\tau^{\frac{4}{3}}|x|^{\frac{4}{3}}\right)\D\nu_1(x).
\end{split}
\end{equation}
amount non-negative Borel measures $\nu_1$, $\nu_2$ and $\nu_3$ that
satisfy the following properties.
\begin{enumerate}
\item All the measures $\nu_j$, $j=1,2,3$ have finite logarithmic
energies;
\item $\nu_1$ and $\nu_3$ are measures supported on $\mathbb{R}$
with $\nu_1(\mathbb{R})=1$ and $\nu_3(\mathbb{R})=\frac{1}{3}$;
\item $\nu_2$ is a measure supported on $i\mathbb{R}$ with
$\nu_2(i\mathbb{R})=\frac{2}{3}$;
\item Let $\sigma$ be the unbounded measure on $i\mathbb{R}$ given
by
\begin{equation}\label{eq:sigma}
\D\sigma(z)=\frac{\sqrt{3}}{2\pi}\tau^{\frac{4}{3}}|z|^{\frac{1}{3}}|\D
z|,\quad z\in i\mathbb{R},
\end{equation}
where $|\D z|$ is the unit arc length on $i\mathbb{R}$, then $\nu_2$
satisfies the constraint $\nu_2\leq\sigma$.
\end{enumerate}
\end{definition}
Let $U^{\nu}(x)$ be the logarithmic potential of the measure $\nu$.
\begin{equation}\label{eq:logpot}
U^{\nu}(x)=-\int\log|x-s|\D\nu(s),
\end{equation}
then it was shown in \cite{DK} that the
logarithmic potentials of $\mu_1$ and $\mu_2$ satisfy the following
properties
\begin{equation}\label{eq:reg}
\begin{split}
2U^{\mu_1}(x)&=
U^{\mu_2}(x)-V(x)+\frac{3}{4}\tau^{\frac{4}{3}}|x|^{\frac{4}{3}}+l,\quad
x\in S_{\mu_1},\\
2U^{\mu_1}(x)&\geq
U^{\mu_2}(x)-V(x)+\frac{3}{4}\tau^{\frac{4}{3}}|x|^{\frac{4}{3}}+l,\quad
x\in\mathbb{R}\setminus S_{\mu_1},
\end{split}
\end{equation}
for some constant $l$, where $S_{\mu_1}$ is the support of $\mu_1$.
The properties of the equilibrium measures $\mu_1$, $\mu_2$ and
$\mu_3$ were studied in \cite{DK} and we have the following
\begin{theorem}\label{thm:equ} (Theorem 2.3 in \cite{DK} (see also Theorem 5.2.2 in \cite{Du}))
Let $V$ be an even polynomial and $\tau>0$. Then there is a unique
minimizer $(\mu_1,\mu_2,\mu_3)$ of $E(\nu_1,\nu_2,\nu_3)$ in
(\ref{eq:EV}) that satisfies the conditions in Definition
\ref{de:equ}. Let us denote the support of the Borel measure $\nu$
by $S_{\nu}$, then we have
\begin{enumerate}
\item $S_{\mu_1}$ consists of finitely many disjoint intervals
\begin{equation}\label{eq:Smu1}
\begin{split}
S_{\mu_1}=\cup_{j=1}^{g+1}[\lambda_{2j-1},\lambda_{2j}],
\end{split}
\end{equation}
where $\lambda_j\in\mathbb{R}$ and the points are ordered such that
$\lambda_j<\lambda_k$ if $j<k$. Moreover, $\mu_1$ is absolutely
continuous with respect to the Lebesgue measure, and on each
interval $[\lambda_{2j-1},\lambda_{2j}]$, it has a continuous
density of the form
\begin{equation}\label{eq:mu1den}
\mu_1(z)=\rho_1(z)\D
z=\psi_j(z)\sqrt{(\lambda_{2j}-z)(z-\lambda_{2j-1})},\quad
z\in[\lambda_{2j-1},\lambda_{2j}],
\end{equation}
where $\psi_j(z)$ is analytic and non-negative on
$[\lambda_{2j-1},\lambda_{2j}]$.
\item Let $\sigma$ be the measure in (\ref{eq:sigma}), then
$S_{\mu_2}=i\mathbb{R}$ and there exists a constant $c>0$
such that the support $S_{\sigma-\mu_2}$ of $\sigma-\mu_2$ is given
by
\begin{equation}\label{eq:suppmu2}
S_{\sigma-\mu_2}=i\mathbb{R}\setminus(-ic,ic).
\end{equation}
Moreover, $\sigma-\mu_2$ has an analytic density on
$S_{\sigma-\mu_2}$ that vanishes as a square root at $\pm ic$.
\item $S_{\mu_3}=\mathbb{R}$ and $\mu_3$ has a density which is
analytic in $\mathbb{R}\setminus\{0\}$.
\item For $j=1,2,3$, we have $\mu_j(A)=\mu_j(-A)$ for any Borel set
$A$.
\end{enumerate}
\end{theorem}
\begin{remark}\label{re:sim}
In particular, by 4. in the above, we see that $S_{\mu_1}$ is
symmetric under the map $z\mapsto -z$, that is, we have
$\lambda_k=-\lambda_{2g+2-k+1}$.
\end{remark}

We then have the following definition of regularity. (See Definition 2.5 in \cite{DK})
\begin{definition}\label{de:reg}
The potential $V(x)$ is regular if the following conditions are
satisfied.
\begin{enumerate}
\item The inequality in (\ref{eq:reg}) is strict outside of
$S_{\mu_1}$;
\item The density $\rho_1(z)$ vanishes like a square-root at the end
points of $S_{\mu_1}$;
\item The density $\rho_1(z)$ does not vanish in the interior of
$S_{\mu_1}$.
\end{enumerate}
\end{definition}
It is known that a generic potential $V(x)$ is regular \cite{DK}.
The ``one-cut regular case" is the case when $S_{\mu_1}$ consists
only of a single interval and that $V(x)$ is regular. In \cite{DK},
rigorous asymptotics of biorthogonal polynomials was obtained for
this case. The asymptotics of the biorthogonal polynomials were then
used to obtain an asymptotic expression for the kernel $K_{11}^{n}$
in (\ref{eq:ker}). In this paper we will extend these result to the
case when $\mu_1$ is supported on any number of intervals. (See
Theorem \ref{thm:main1} and Theorem \ref{thm:main2})

\section{Statement of results}\label{se:Re}
In this paper we obtain universality results for the 2 matrix model
with potentials $W(y)=\frac{y^4}{4}$ and $V(x)$ a general even
polynomial. Moreover, we will assume the potential $V(x)$ satisfies
the regularity condition in Definition \ref{de:reg} and that $n$ is
a multiple of $3$. Then we have the following result on the global
eigenvalue distribution of the matrix $M_1$.
\begin{theorem}\label{thm:main1}
Let $(\mu_1,\mu_2,\mu_3)$ be the equilibrium measures that minimize
the functional (\ref{eq:EV}). Then as $n\rightarrow\infty$ and
$n\equiv 0(\mathrm{mod}3)$, we have
\begin{equation}\label{eq:global}
\lim_{n\rightarrow\infty}\frac{1}{n}K_{11}^{n}(x,x)=\rho_1(x),
\end{equation}
uniformly for $x\in\mathbb{R}$, where $\rho_1$ is the density of
$\mu_1$ in Definition \ref{de:reg}.
\end{theorem}
As explained in \cite{DK}, the requirement $n\equiv
0(\mathrm{mod}3)$ is not essential and is only imposed to minimize
the technicality. The other result is the universality property of
the kernel $K_{11}^{n}$.
\begin{theorem}\label{thm:main2}
Let $(\mu_1,\mu_2,\mu_3)$ be the equilibrium measures that minimize
the functional (\ref{eq:EV}). Then as $n\rightarrow\infty$ and
$n\equiv 0(\mathrm{mod}3)$, we have the followings.
\begin{enumerate}
\item Let $x^{\ast}$ be a point in the interior of the support
$S_{\mu_1}$. Then we have
\begin{equation}\label{eq:bulk}
\begin{split}
\lim_{n\rightarrow\infty}&\frac{1}{n\rho_1(x^{\ast})}K_{11}^{n}\left(x^{\ast}+\frac{u}{n\rho_1(x^{\ast})}
,x^{\ast}+\frac{v}{n\rho_1(x^{\ast})}\right)=\frac{\sin\left(\pi(u-v)\right)}{\pi(u-v)},
\end{split}
\end{equation}
uniformly for $u$, $v$ in compact subsets of $\mathbb{R}$.
\item Let $\varphi_j>0$ be such that
\begin{equation*}
\rho_1(z)=\frac{\varphi_j}{\pi}|z-\lambda_j|^{\frac{1}{2}}+O\left(z-\lambda_j\right),
\end{equation*}
as $z\rightarrow\lambda_j$, $j=1,\ldots,2g+2$ inside of $S_{\mu_1}$,
where $\lambda_j$ are defined as in (\ref{eq:Smu1}). Then we have
the following
\begin{equation}\label{eq:edge}
\begin{split}
\lim_{n\rightarrow\infty}&\frac{1}{\left(n\varphi_j\right)^{\frac{2}{3}}}K_{11}^{n}\left(\lambda_j+(-1)^j\frac{u}{\left(n\varphi_j\right)^{\frac{2}{3}}}
,\lambda_j+(-1)^j\frac{v}{\left(n\varphi_j\right)^{\frac{2}{3}}}\right)\\
&=\frac{\mathrm{Ai}(u)\mathrm{Ai}^{\prime}(v)-\mathrm{Ai}(v)\mathrm{Ai}^{\prime}(u)}{u-v},
\end{split}
\end{equation}
uniformly for $u$, $v$ in compact subsets of $\mathbb{R}$, where
$\mathrm{Ai}$ is the Airy function.
\end{enumerate}
\end{theorem}
Recall that the Airy function is the unique solution to the
differential equation $v^{\prime\prime}=zv$ that has the following
asymptotic behavior as $z\rightarrow\infty$ in the sector
$-\pi+\epsilon\leq \arg(z)\leq \pi-\epsilon$, for any $\epsilon>0$.
\begin{equation}\label{eq:asymairy}
\mathrm{Ai}(z)=\frac{1}{2\sqrt{\pi}z^{\frac{1}{4}}}e^{-\frac{2}{3}z^{\frac{3}{2}}}\left(1+O(z^{-\frac{3}{2}})\right)
,\quad -\pi+\epsilon\leq \arg(z)\leq \pi-\epsilon,\quad
z\rightarrow\infty.
\end{equation}
where the branch cut of $z^\frac{3}{2}$ in the above is chosen to be
the negative real axis.

Although the steepest decent analysis in \cite{DK} already covers
the general case without the 1-cut assumption, solution to a
`modeled Riemann-Hilbert problem' (See (\ref{eq:outer})) must be
obtained to complete the Riemann-Hilbert analysis and to extend the
universality results to the general case. The main difficulties are
to show that a solution of the modeled Riemann-Hilbert problem
exists for all $n$ and to find an explicit expression of it. This
involves the study of the theta divisor, which is the set of points
in which a theta function vanishes. (See Section \ref{se:theta} for
a more detailed description of the theta function) This is a
difficult problem with very few results available. In this paper we
managed to construct the solution of the modeled Riemann-Hilbert
problem (\ref{eq:outer}) with the use of theta functions and by
using results from \cite{FK} and \cite{Fay}, we were able to show
the existence of the solution $M(z)$ to (\ref{eq:outer}) for all
$n$. This allows us to extend the universality results in \cite{DK}
to the case when $V(x)$ is a general even polynomial and obtain
Theorem \ref{thm:main1} and Theorem \ref{thm:main2}.

In many
applications of the Deift-Zhou steepest decent method, theta
function is needed to solve a modeled Riemann-Hilbert problem and
the solvability of these modeled
Riemann-Hilbert problems is important to guarantee the
validity of these asymptotic formula for all $n$, as
$n\rightarrow\infty$. We believe the techniques and results in this
paper will not only be valuable to the random matrix community
studying 2 matrix models, but it will also be important to many
other problems in which the Deift-Zhou steepest decent method is
applicable.

\section*{Acknowledgement}
The author is indebted to M. Duits and A. B. J. Kuijlaars for many
fruitful discussions and for inviting me to Katholieke Universiteit
Leuven in which this work initiated. I am also grateful to them for
providing me with an early copy of \cite{Du} and for encouraging me
to complete this paper. I am also grateful to M. Bertola for showing
me the relevant results in \cite{Fay} which led to the completion of
Section \ref{se:nonvan}. The funding of this research is provided by
the EPSRC grant EP/D505534/1.

\section{Riemann-Hilbert analysis}\label{se:RHP}
In this section we will summarize the results in \cite{DK} that are
relevant to our analysis. We will not repeat the analysis in
\cite{DK} but will simply state the results that are applicable to
our studies.

\subsection{Riemann-Hilbert problem}
The biorthogonal polynomials $p_k(x)$ in (\ref{eq:bop}) satisfies a
Riemann-Hilbert problem \cite{BEHRH}, \cite{KMRH} similar to the one
that is satisfied by orthogonal polynomials \cite{FI}. This allows
the implementation of the Deift-Zhou steepest decent method.
(\cite{D}, \cite{DKV}, \cite{DKV2}, \cite{BI})

Let $w_j(x)$ be the weights defined by
\begin{equation}\label{eq:weight}
w_j(x)=e^{-nV(x)}\int_{\mathbb{R}}y^je^{-n\left(\frac{y^4}{4}-\tau
xy\right)}\D y,\quad j=0,1,2.
\end{equation}
Assuming $n$ is divisible by $3$ and consider the following
Riemann-Hilbert problem
\begin{equation}\label{eq:RHP}
\begin{split}
1.\quad &\text{$Y(z)$ is analytic in
$\mathbb{C}\setminus\mathbb{R}$},\\
2.\quad &Y_+(z)=Y_-(z)\begin{pmatrix}1&w_0(z)&w_1(z)&w_2(z)\\
0&1&0&0\\
0&0&1&0\\
0&0&0&1
\end{pmatrix},\quad z\in\mathbb{R}\\
3.\quad &Y(z)=\left(I+O(z^{-1})\right)\begin{pmatrix}z^{n}&0&0&0\\
0&z^{-\frac{n}{3}}&0&0\\
0&0&z^{-\frac{n}{3}}&0\\
0&0&0&z^{-\frac{n}{3}}
\end{pmatrix},\quad z\rightarrow\infty.
\end{split}
\end{equation}
This Riemann-Hilbert problem has a unique solution given by the
biorthogonal polynomial $p_k(x)$ and some other polynomials,
together with their Cauchy transforms \cite{KMRH}.
\begin{equation}\label{eq:RHsol}
Y(z)=\begin{pmatrix}p_n(z)&C(p_nw_0)(z)&C(p_nw_1)(x)&C(p_nw_2)(z)\\
p_{n-1}^{(0)}(z)&C(p_{n-1}^{(0)}w_0)(z)&C(p_{n-1}^{(0)}w_1)(z)&C(p_{n-1}^{(0)}w_2)(z)\\
p_{n-1}^{(1)}(z)&C(p_{n-1}^{(1)}w_0)(z)&C(p_{n-1}^{(1)}w_1)(z)&C(p_{n-1}^{(1)}w_2)(z)\\
p_{n-1}^{(2)}(z)&C(p_{n-1}^{(2)}w_0)(z)&C(p_{n-1}^{(2)}w_1)(z)&C(p_{n-1}^{(2)}w_2)(z)
\end{pmatrix},
\end{equation}
where $p_{n-1}^{(j)}(z)$, $j=0,1,2$ are some polynomials of degree
$n-1$ and $C(f)$ is the Cauchy transform of the function $f$.
\begin{equation*}
C(f)(x)=\frac{1}{2\pi i}\int_{\mathbb{R}}\frac{f(s)}{s-x}\D s.
\end{equation*}
In \cite{DK}, the Deift-Zhou steepest decent method (\cite{D},
\cite{DKV}, \cite{DKV2}, \cite{BI}) was extended to the
Riemann-Hilbert problem (\ref{eq:RHsol}). With the help of the
equilibrium measures $\mu_1$, $\mu_2$ and $\mu_3$ that minimize
(\ref{eq:EV}), the authors of \cite{DK} were able to transform and approximate
the Riemann-Hilbert problem (\ref{eq:RHsol}) by an explicitly
solvable one. To state their results, let us assume $V(x)$ is
regular. Let us denote an interval in the support $S_{\mu_1}$
(\ref{eq:Smu1}) of $\mu_1$ by $\Xi_j$ and a gap between the
intervals by $\tilde{\Xi}_j$.
\begin{equation}\label{eq:cutgaps}
\begin{split}
\Xi_j&=[\lambda_{2j-1},\lambda_{2j}],\quad j=1,\ldots, g+1,\\
\tilde{\Xi}_{j}&=[\lambda_{2j},\lambda_{2j+1}],\quad j=1,\ldots, g,\\
\tilde{\Xi}_0&=(-\infty,\lambda_1],\quad
\tilde{\Xi}_{g+1}=[\lambda_{2g+2},\infty).
\end{split}
\end{equation}
We will define $\alpha_j$ to be the constants
\begin{equation}\label{eq:apj}
\alpha_{k}=\mu_1\left(\cup_{j=k+1}^{g+1}\Xi_{j}\right),\quad
k=0,\ldots, g,\quad \alpha_{g+1}=0.
\end{equation}
Note that, since $V(x)$ is an even polynomial, by Theorem
\ref{thm:equ}, we have $\mu_1(A)=\mu_1(-A)$ for any Borel set $A$.
Therefore the constants $\alpha_k$ in (\ref{eq:apj}) satisfy the
symmetry
\begin{equation}\label{eq:symmetry}
\alpha_k=1-\alpha_{g+1-k}.
\end{equation}
Let us define the following Riemann-Hilbert problem for a matrix
$M(z)$. (See Section 8 of \cite{DK} and 5.10 of \cite{Du})
\begin{equation}\label{eq:outer}
\begin{split}
1.&\quad \textrm{$M(z)$ is analytic in
$\mathbb{C}\setminus\left(\mathbb{R}\cup S_{\sigma-\mu_2}\right)$,}\\
2.&\quad M_+(z)=M_-(z)J_M(z),\quad z\in\mathbb{R}\cup
S_{\sigma-\mu_2},\\
3.& \quad M(z)=\left(I+O(z^{-1})\right)\begin{pmatrix}1&0&0&0\\
0&z^{\frac{1}{3}}&0&0\\
0&0&1&0\\
0&0&0&z^{-\frac{1}{3}}\end{pmatrix}\begin{pmatrix}1&0\\
0&A_j\end{pmatrix},
\\
&\textrm{uniformly as $z\rightarrow\infty$ in the $j^{th}$
quadrant,}
\\
4.&\quad M(z)=O\left((z-\lambda_j)^{-\frac{1}{4}}\right),\quad
z\rightarrow\lambda_j,\quad j=1,\ldots, 2g+2,\\
&\quad M(z)=O\left((z\mp ic)^{-\frac{1}{4}}\right),\quad
z\rightarrow\pm ic.
\end{split}
\end{equation}
where $S_{\sigma-\mu_2}$ is oriented upward, the branch of
$z^{\frac{1}{3}}$ is chosen such that $z^{\frac{1}{3}}\in\mathbb{R}$
for $z\in\mathbb{R}_+$ and the branch cut is chosen to be the
negative real axis. The $A_j$ are given by (with
$\omega=e^{\frac{2\pi i}{3}}$)
\begin{equation*}
\begin{split}
A_1&=\frac{i}{\sqrt{3}}\begin{pmatrix}-1&\omega&\omega^2\\
                                      -1&1&1\\
                                      -1&\omega^2&\omega
                                      \end{pmatrix},\quad
A_2=\frac{i}{\sqrt{3}}\begin{pmatrix}\omega&1&\omega^2\\
                                      1&1&1\\
                                      \omega^2&1&\omega
                                      \end{pmatrix},\\
A_3&=\frac{i}{\sqrt{3}}\begin{pmatrix}\omega^2&1&-\omega\\
                                      1&1&-1\\
                                      \omega&1&-\omega^2
                                      \end{pmatrix},\quad
A_4=\frac{i}{\sqrt{3}}\begin{pmatrix}-1&\omega^2&-\omega\\
                                      -1&1&-1\\
                                      -1&\omega&-\omega^2
                                      \end{pmatrix}.
\end{split}
\end{equation*}
The jump matrices $J_M(z)$ in (\ref{eq:outer}) are given by the
followings
\begin{equation}\label{eq:JM}
\begin{split}
J_M(z)&=\begin{pmatrix}0&1&0&0\\
-1&0&0&0\\0&0&0&1\\0&0&-1&0\end{pmatrix}
 ,\quad z\in S_{\mu_1},\\
J_M(z)&=\begin{pmatrix}e^{-2n\pi i\alpha_k}&0&0&0\\
0&e^{2n\pi i\alpha_k}&0&0\\
0&0&0&1\\0&0&-1&0\end{pmatrix},\quad z\in\tilde{\Xi}_k,\quad
k=0,\ldots,g+1,\\
J_M(z)&=\begin{pmatrix}1&0&0&0\\
0&0&-1&0\\
0&1&0&0\\
0&0&0&1\end{pmatrix},\quad z\in S_{\sigma-\mu_2}.
\end{split}
\end{equation}
The steepest decent analysis in \cite{DK} leads to the `modeled
Riemann-Hilbert problem' (\ref{eq:outer}). Provided a solution
$M(z)$ of (\ref{eq:outer}) exists and is uniformly bounded in $n$
away from the singularities, the analysis in \cite{DK} that leads to
the asymptotic forms (\ref{eq:global}), (\ref{eq:bulk}) and
(\ref{eq:edge}) can be carried out with the parametrix $M(z)$.
\begin{theorem}\label{thm:DK}
Let $\varepsilon>0$ be a fixed small number independent on $n$. Let
$B_{\varepsilon,j}$ and $B_{\varepsilon,\pm ic}$ be small discs of
radius $\varepsilon$ centered at $\lambda_j$ and $\pm ic$
respectively. Let $\mathcal{K}\subset\mathbb{C}$ be a compact subset
in $\mathbb{C}$ and let $\mathcal{T}$ be the set
\begin{equation}\label{eq:calT}
\mathcal{T}=\mathcal{K}\setminus\left(\cup_{j=1}^{2g+2}B_{\varepsilon,j}\cup
B_{\varepsilon,ic}\cup B_{\varepsilon,-ic}\right).
\end{equation}
Suppose the solution $M(z)$ of the Riemann-Hilbert problem
(\ref{eq:outer}) and its inverse $M^{-1}(z)$ exist and satisfy the
following conditions.
\begin{enumerate}
\item Both $M(z)$ and $M^{-1}(z)$ are bounded in $n$ uniformly inside $\mathcal{T}$ for any compact
subset $\mathcal{K}$;
\item For any $r>\mathrm{max}\{c,\lambda_{2g+2}\}$ independent on $n$, there exist constants $C_{jk}^l$
and $\varpi_{jk}^l$, $1\leq j,k,l\leq 4$, independent on $n$, such
that, for $z>|r|$,
\begin{equation*}
\begin{split}
&|\left(M(z)\right)_{jk}|<C_{jk}^l\left|z^{\frac{1}{3}}\right|,\quad
\textrm{for $z$ in the $l^{th}$
quadrant,}\\
&|\left(M^{-1}(z)\right)_{jk}|<\varpi_{jk}^l\left|z^{\frac{1}{3}}\right|,\quad
\textrm{for $z$ in the $l^{th}$ quadrant.}
\end{split}
\end{equation*}
\end{enumerate}
Then as $n\rightarrow\infty$ and $n\equiv 0\mathrm{mod}3$, the
asymptotic behavior of the kernel $K_{11}^{n}$ is given by
(\ref{eq:global}), (\ref{eq:bulk}) and (\ref{eq:edge}).
\end{theorem}
In the following sections we will construct the solution $M(z)$ with
the help of theta functions and we will show that the solution
satisfies the conditions in Theorem \ref{thm:DK} in Section
\ref{se:nonvan}.

\section{Theta function and Riemann surface}\label{se:theta}
We will now construct a Riemann surface from the equilibrium
measures and use the theta function on this Riemann surface to
construct a parametrix for the Riemann-Hilbert problem
(\ref{eq:outer}).

The Riemann surface is realized as a four-sheeted covering of the
Riemann sphere. Define four copies of the Riemann sphere by
$\mathcal{\Lie}_j$, $j=1,\ldots,4$
\begin{equation}
\begin{split}
\mathcal{\Lie}_1&=\overline{\mathbb{C}}\setminus S_{\mu_1},\quad
\mathcal{\Lie}_2=\overline{\mathbb{C}}\setminus (S_{\mu_1}\cup
S_{\sigma-\mu_2}),\\
\mathcal{\Lie}_3&=\overline{\mathbb{C}}\setminus
(S_{\sigma-\mu_2}\cup S_{\mu_3}),\quad
\mathcal{\Lie}_4=\overline{\mathbb{C}}\setminus (S_{\mu_3}),
\end{split}
\end{equation}
where $\overline{\mathbb{C}}$ is the Riemann sphere obtained by
adding the point $z=\infty$ to $\mathbb{C}$.

The Riemann surface $\mathcal{\Lie}$ is constructed as follows:
$\mathcal{\Lie}_1$ is connected to $\mathcal{\Lie}_2$ via
$S_{\mu_1}$, $\mathcal{\Lie}_2$ is connected to $\mathcal{\Lie}_3$
via $S_{\sigma-\mu_2}$ and $\mathcal{\Lie}_3$ is connected to
$\mathcal{\Lie}_4$ via $S_{\mu_3}$, as shown in Figure
\ref{fig:connection}.
\begin{figure}
\centering
\includegraphics[scale=0.75]{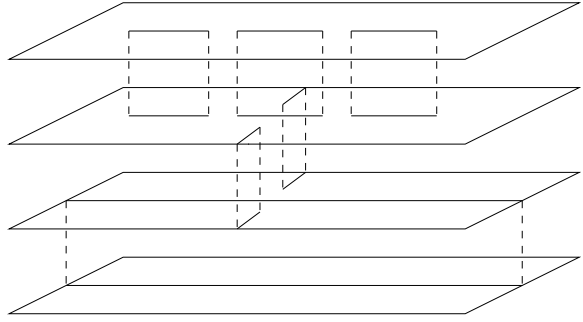} \caption{The sheet structure of the Riemann surface
$\mathcal{\Lie}$.}\label{fig:connection}
\end{figure}
Let us define the functions $F_j(z)$ by
\begin{equation}
F_j(z)=\int_{S_{\mu_j}}\frac{1}{z-s}\D\mu_j(s).
\end{equation}
Then we have the following result
\begin{lemma}\label{le:RiemSur} (Lemma 5.1 in
\cite{DK}, (see also Lemma 5.4.1 in \cite{Du})) The function
$\xi:\cup_{j=1}^4\mathcal{\Lie}_j\rightarrow\overline{\mathbb{C}}$
defined by
\begin{equation}
\xi(z)=\left\{
         \begin{array}{ll}
           -F_1(z)+V^{\prime}(z), & \hbox{$z\in\mathcal{\Lie}_1$;} \\
           F_1(z)-F_2(z)+\tau^{\frac{4}{3}}z^{\frac{1}{3}}, & \hbox{$z\in\mathcal{\Lie}_2,\quad\mathrm{Re}z>0$;} \\
           F_1(z)-F_2(z)-\tau^{\frac{4}{3}}(-z)^{\frac{1}{3}}, & \hbox{$z\in\mathcal{\Lie}_2,\quad\mathrm{Re}z<0$;} \\
           F_2(z)-F_3(z)-\tau^{\frac{4}{3}}(-z)^{\frac{1}{3}}, & \hbox{$z\in\mathcal{\Lie}_3,\quad\mathrm{Re}z>0$;} \\
           F_2(z)-F_3(z)+\tau^{\frac{4}{3}}z^{\frac{1}{3}}, & \hbox{$z\in\mathcal{\Lie}_3,\quad\mathrm{Re}z<0$;} \\
           F_3(z)+e^{\frac{4\pi i}{3}}\tau^{\frac{4}{3}}z^{\frac{1}{3}}, & \hbox{$z\in\mathcal{\Lie}_4,\quad\mathrm{Im}z>0$;} \\
           F_3(z)+e^{\frac{2\pi i}{3}}\tau^{\frac{4}{3}}z^{\frac{1}{3}}, & \hbox{$z\in\mathcal{\Lie}_4,\quad\mathrm{Im}z<0$.}
         \end{array}
       \right.
\end{equation}
has an extension to a meromorphic function (also denoted by $\xi$)
on $\mathcal{\Lie}$. The meromorphic function has a pole of order
$\mathrm{deg}V-1$ at infinity on the first sheet, and a simple pole
at the other points at infinity. We shall denote the restriction of
$\xi(z)$ to the sheet $\mathcal{\Lie}_k$ by $\xi_k(z)$.
\end{lemma}
The Riemann surface $\mathcal{\Lie}$ is of genus $g$. Let us define
a set of canonical basis of cycle as in Figure \ref{fig:cycle}.
\begin{figure}
\centering \psfrag{lambda1}[][][1][0.0]{\tiny$\lambda_1$}
\psfrag{l2}[][][1][0.0]{\tiny$\lambda_2$}
\psfrag{l3}[][][1][0.0]{\tiny$\lambda_3$}
\psfrag{l4}[][][1][0.0]{\tiny$\lambda_4$}
\psfrag{l2g}[][][1][0.0]{\tiny$\lambda_{2g+2}$}
\psfrag{a1}[][][1][0.0]{\tiny$b_g$}
\psfrag{a2}[][][1][0.0]{\tiny$b_{g-1}$}
\psfrag{a3}[][][1][0.0]{\tiny$b_1$}
\psfrag{b1}[][][1][0.0]{\tiny$a_1$}
\psfrag{b2}[][][1][0.0]{\tiny$a_{g-1}$}
\psfrag{b3}[][][1][0.0]{\tiny$a_g$}
\psfrag{it}[][][1][0.0]{\tiny$ic$}
\psfrag{-it}[][][1][0.0]{\tiny$-ic$}
\includegraphics[scale=0.75]{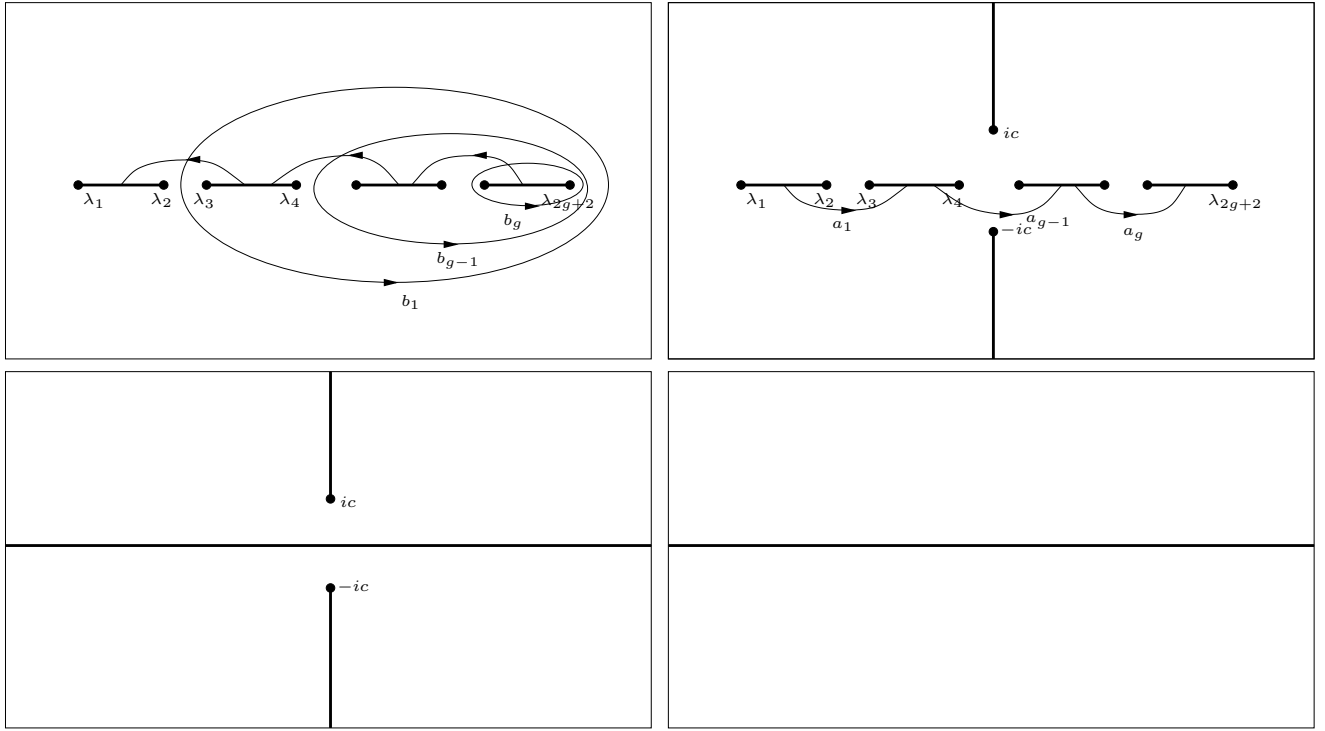}
\caption{The $a$ and $b$ cycle of Riemann surface
$\mathcal{\Lie}$.}\label{fig:cycle}
\end{figure}
The figure should be understood as follows. The top left rectangle
denotes the first sheet $\mathcal{\Lie}_1$, the top right rectangle
denotes $\mathcal{\Lie}_2$, the lower left one denotes
$\mathcal{\Lie}_3$ and the lower right one denotes
$\mathcal{\Lie}_4$. A $b$-cycle is a loop in $\mathcal{\Lie}_1$
around the branch cuts that is symmetric with respect to the real
axis, while an $a$-cycle $a_j$ consist of a path in the upper half
plane in $\mathcal{\Lie}_1$ that goes from $\Xi_{j+1}$ to $\Xi_{j}$
($\Xi_j$ is defined in (\ref{eq:cutgaps})), together with a path in
the lower half plane in $\mathcal{\Lie}_2$ that goes from $\Xi_{j}$
to $\Xi_{j+1}$. The loop formed by these 2 paths is an $a$-cycle. We
will also choose these 2 paths such that their projection on the
complex $z$-plane are mapped onto each other under complex
conjugation.

We can now define the basis of holomorphic differential that is dual
to this basis of cycle.

Let $\D\omega_j$ be holomorphic differential 1-forms on
$\mathcal{\Lie}$ such that
\begin{equation}
\oint_{a_k}\D\omega_j=\delta_{jk},\quad 1\leq j,k\leq g.
\end{equation}
The 1-forms $\D\omega_j$ are known as the holomorphic 1-forms that
are dual to the basis of cycles $(a,b)$.

Let the $b$-period of these 1-forms be $\Pi_{ij}$
\begin{equation}\label{eq:permatrix}
\oint_{b_i}\D\omega_j=\Pi_{ij},
\end{equation}
then the $g\times g$ matrix $\Pi$ with entries $\Pi_{ij}$ is
symmetric and $\mathrm{Im}({\Pi})>0$.
\subsection{Theta function and its properties}

The theta function $\theta:\mathbb{C}^g\longrightarrow \mathbb{C}$
associated to the Riemann surface $\mathcal{\Lie}$ and this choice
of basis is defined by
\begin{equation}
  \label{eq:thetadef}
  \theta (\vec{s}) := \sum_{\vec{n}\in
    \mathbb{Z}^g} {\rm e}^{i\pi \vec{n}\cdot \Pi
    \vec{n} + 2i\pi \vec{s}\cdot
    \vec{n}}.
\end{equation}
The theta function has the following quasi-periodic property, which
will be important to the construction of the parametrix.
\begin{proposition}
\label{pro:per} The theta function is quasi-periodic with the
following properties:
\begin{equation}
  \label{eq:period}
\begin{split}
  \theta (\vec{s}+\vec{M})=&
\theta(\vec{s}), \\
\theta (\vec{s}+\Pi\vec{M})=& \exp \left[ 2\pi
i\left(-\left<\vec{M},\vec{s}\right>-\left<\vec{M},{\Pi\over{2}}\vec{M}\right>\right)\right]
\theta(\vec{s}),
\end{split}
\end{equation}
\end{proposition}
where $ \left \langle \cdot, \cdot \right \rangle$ denotes the usual
inner product in $\mathbb{C}^g$.

We will now define the Abel map on $\mathcal{\Lie}$. The Abel map
$u: \mathcal{\Lie}\rightarrow\mathbb{C}^g$ is defined by
\begin{equation}\label{eq:abel}
  u(x)=\left(u_1(x),\ldots,u_g(x)\right)^T=\left(\int_{x_0}^x\D\omega_1,\ldots,\int_{x_0}^x\D\omega_g\right)^T,
\end{equation}
where $x_0$ is a point on $\mathcal{\Lie}$. We will choose $x_0$ so
that $x_0$ is the point on $\mathcal{\Lie}_1$ that projects to
$\lambda_{2g+2}$ in $\overline{\mathbb{C}}$. We will denote this
point by $\lambda_{2g+2}^1$.

The composition of the theta function with the Abel map is then a
multi-valued function from $\mathcal{\Lie}$ to $\mathbb{C}$. It is
either identically zero or it has $g$ zeros on $\mathcal{\Lie}$. The
following lemma tells us where the zeros are.
\begin{lemma}
\label{le:zero} Let $D=\sum_{i=1}^gd_i$ be a non special divisor of
degree $g$ on $\mathcal{\Lie}$, then the multi-valued function
\[
  \theta (u(x)-u(D)-\vec{K})
\]
has precisely $g$ zeros located at the points $d_i$, $i=1,\ldots,g$.
The vector $\vec{K}=(K_1,\ldots,K_g)^T$ is the Riemann constant
\[
  K_j={{\Pi_{jj}}\over 2}-\sum_{l=
  1}^g\int_{a_l}(\D\omega_l(x)\int_{\lambda_{2g+2}^1}^x\D\omega_j).
\]
\end{lemma}
Recall that a divisor of degree $m$ is a formal sum of $m$ points
(counting multiplicity) on the Riemann surface and that two divisors
$D_1$ and $D_2$ are equivalent if and only if there exists a
meromorphic function $f(x)$ on $\mathcal{\Lie}$ with poles at the
points of $D_1$ and zeros at the points of $D_2$. A divisor
$\sum_{i=1}^gd_i$ is special if there exists a non-constant
meromorphic function on $\mathcal{\Lie}$ with $g$ poles at the
points $d_1,\ldots,d_g$. The condition of $D$ being non special is
equivalent to the condition that $\theta (u(x)-u(D)-\vec{K})$ does
not vanish identically.
\begin{theorem}\label{thm:RR} Let $d_1,\ldots,d_g$ be $g$ points on a Riemann
surface and let the multiplicity of $d_j$ within these $g$-tuple of points be $k_j$.
Then the function $\theta(u(z)-\sum_{j=1}^gu(d_g)-K)$ is
identically zero if and only if there exist a function $f(z)$
that has poles of order $k_j$ at $d_j$ for $j=1,\ldots,g$ and holomorphic elsewhere.
\end{theorem}
This is a consequence of the Riemann-Roch theorem. In general, for a
given $g+l$ points (counting multiplicity)  on a Riemann surface,
there are $l$ independent meromorphic functions with poles exactly
at these points. This can be thought of as an extension of the
Liouville's theorem.

Let $\phi(z)$ be the anti-holomorphic involution on $\mathcal{\Lie}$
defined by
\begin{equation}\label{eq:phi}
\phi(z):(z,\xi(z))\rightarrow
\left(\overline{z},\xi\left(\overline{z}\right)\right)
\end{equation}
where $\xi(z)$ is the function on $\mathcal{\Lie}$ given by Lemma
\ref{le:RiemSur}.

Then by the definition of the cycles in Figure \ref{fig:cycle}, we
see that under the involution $\phi$, we have
\begin{equation}\label{eq:involcyc}
\phi(b_j)=-b_j,\quad \phi(a_j)\sim a_j,\quad j=1,\ldots, g,
\end{equation}
where the symbol $\sim$ means that $\phi(a_j)$ is homologic to
$a_j$.

In particular, if we consider the holomorphic 1-forms
$\overline{\D\omega_j(\phi(x))}$ on $\mathcal{\Lie}$, we have
\begin{equation*}
\begin{split}
\oint_{a_k}\overline{\D\omega_j(\phi(x))}&=\oint_{\phi(a_k)}\overline{\D\omega_j(x)}=
\oint_{a_k}\overline{\D\omega_j(x)}=\delta_{jk}.
\end{split}
\end{equation*}
Hence by the uniqueness of holomorphic 1-forms that is dual to the
cycles $(a,b)$, we have
$\overline{\D\omega_j(\phi(x))}=\omega_j(x)$. By computing the
$b$-periods of $\overline{\D\omega_j(\phi(x))}$ and making use of
(\ref{eq:involcyc}), we obtain the following.
\begin{lemma}\label{le:imag}
The period matrix $\Pi$ of $\mathcal{\Lie}$ is purely imaginary.
\end{lemma}

In particular, by (\ref{eq:thetadef}), we see that $\theta(0)$ is
real and positive and from Lemma \ref{le:zero}, we see that
$\theta(u(x))$ has $g$ zeros on $\mathcal{\Lie}$. Let us denote
these zeros by $\iota_1,\ldots,\iota_g$. Then by the following
result in \cite{FK}, we can simplify the expression of the Riemann
constant $\vec{K}$.
\begin{proposition}\label{pro:constant}
(See p.308-309 of $\cite{FK}$) Suppose $\theta(u(x))$ is not
identically zero. Then it has $g$ zeros in $\mathcal{\Lie}$. Let
$\iota_1,\ldots,\iota_g$ be its zeros, then the Riemann constant is
given by
\begin{equation}
\vec{K}=-\sum_{j=1}^gu(\iota_j).
\end{equation}
\end{proposition}
\begin{remark}\label{re:iota}
 Let $\tilde{\Xi}^l_{j}$ be the intervals in
$\mathcal{\Lie}_l$ that projects onto the gaps $\tilde{\Xi}_j$ in
(\ref{eq:cutgaps}), then as we shall see in Corollary
\ref{pro:fay2}, there is exactly one point $\iota_j$ that belongs to
$\tilde{\Xi}^1_j\cup\tilde{\Xi}^2_j$ for $j=1,\ldots,g$. We shall
label the $\iota_j$ such that
$\iota_j\in\tilde{\Xi}^1_j\cup\tilde{\Xi}^2_j$.
\end{remark}
We would like to express the function $\theta(u(x))$ as a
meromorphic function on $\mathbb{C}$ with jump discontinuities. To
do so, we need to define the contour of integration in the Abel map
(\ref{eq:abel}) such that the integral can be defined without
ambiguity. We will define the contour of integration as in Figure
\ref{fig:contour}.
\begin{figure}
\centering \psfrag{lambda1}[][][1][0.0]{\tiny$\lambda_1$}
\psfrag{l2}[][][1][0.0]{\tiny$\lambda_2$}
\psfrag{l3}[][][1][0.0]{\tiny$\lambda_3$}
\psfrag{l4}[][][1][0.0]{\tiny$\lambda_4$}
\psfrag{l2g}[][][1][0.0]{\tiny$\lambda_{2g+2}$}
\psfrag{r1}[][][1][0.0]{\tiny$\Sigma_1$}
\psfrag{r2}[][][1][0.0]{\tiny$\Sigma_2$}
\psfrag{r3}[][][1][0.0]{\tiny$\Sigma_3$}
\psfrag{r4}[][][1][0.0]{\tiny$\Sigma_4$}
\psfrag{it}[][][1][0.0]{\tiny$ic$}
\psfrag{-it}[][][1][0.0]{\tiny$-ic$}
\psfrag{ia}[][][1][0.0]{\tiny$i\zeta$}
\psfrag{-ia}[][][1][0.0]{\tiny$-i\zeta$}
\includegraphics[scale=0.75]{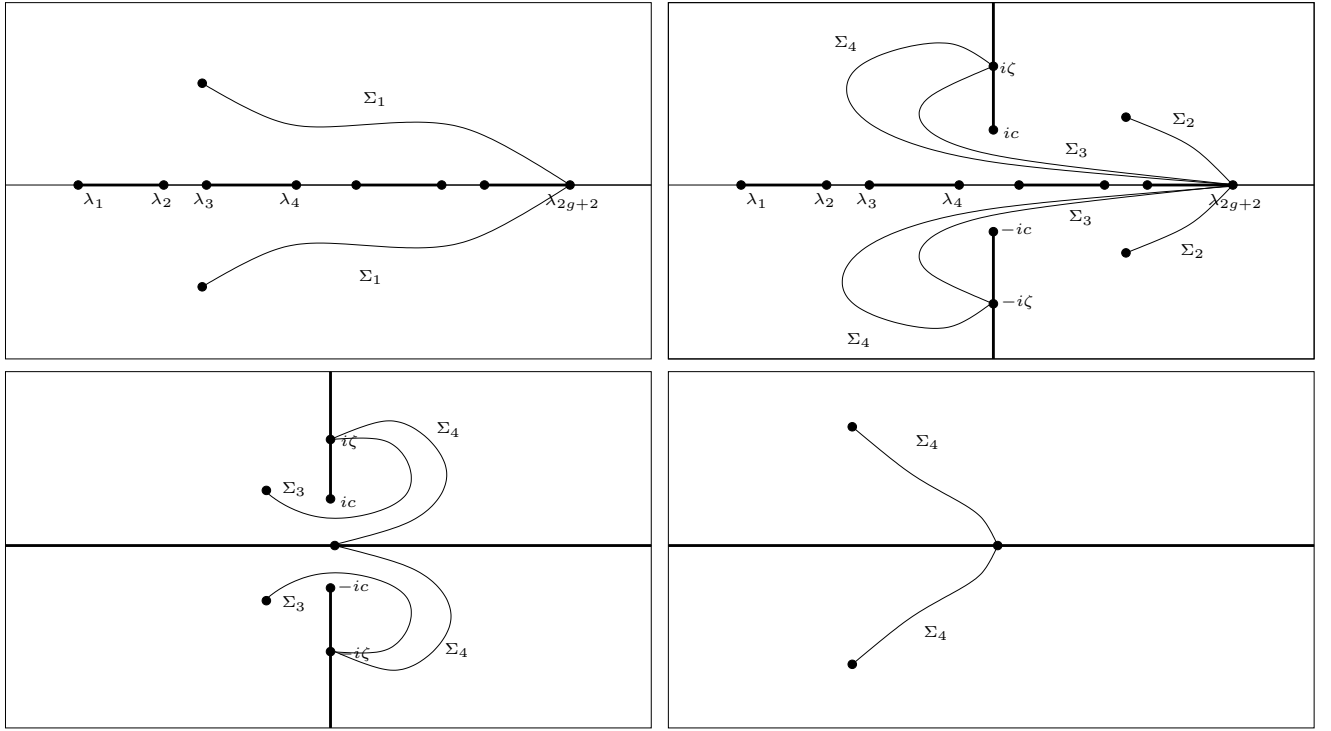}
\caption{The contour of integration for the Abel map
$u$.}\label{fig:contour}
\end{figure}

For a point $z$ in $\mathcal{\Lie}_1$ ($\mathcal{\Lie}_2$), the
contour of integration $\Sigma_1$ goes from $\lambda_{2g+2}$ to $z$
in $\mathcal{\Lie}_1$ ($\mathcal{\Lie}_2$) without intersecting
$(-\infty,\lambda_{2g+2})$ and the branch cuts on the imaginary
axis. For a point $z$ in the upper (lower) half plane in
$\mathcal{\Lie}_3$, the contour of integration $\Sigma_3$ consists
of 2 parts. The first part lies in $\mathcal{\Lie}_2$, goes from
$\lambda_{2g+2}$ to a point $i\zeta$ ($-i\zeta$) on the branch cut
in the imaginary axis without intersecting
$(-\infty,\lambda_{2g+2})$ in $\mathcal{\Lie}_2$ and enter the
branch cut from the left hand side of this point. The second part
lies in the upper (lower) half plane in $\mathcal{\Lie}_3$, goes
from the right hand side of $i\zeta$ ($-i\zeta$) to the point $z$.
For a point $z$ in the upper (lower) half plane in
$\mathcal{\Lie}_4$, the contour of integration consists of 3 parts.
The first part lies in $\mathcal{\Lie}_2$, goes from
$\lambda_{2g+2}$ to a point $-i\zeta$ ($i\zeta$) on the branch cut
in the imaginary axis without intersecting
$(-\infty,\lambda_{2g+2})$ in $\mathcal{\Lie}_2$ and enter the
branch cut from the left hand side of this point. The second part
lies in the lower (upper) half plane in $\mathcal{\Lie}_3$, goes
from the right hand side of $-i\zeta$ ($i\zeta$) to the origin. The
last part lies in the upper (lower) half plane in
$\mathcal{\Lie}_4$, goes from the origin to the point $z$. The
choice of the point $\pm i\zeta$ in the construction is immaterial
as long as it lies on the branch cut on the imaginary axis.

Let $z^j$ be the point on $\mathcal{\Lie}_j$ that projects to $z$ in
$\overline{\mathbb{C}}$ and $\vec{A}$ be a $g$-dimensional constant
vector. We can now define four functions $\theta^j(u(z)+\vec{A})$ on
the complex $z$-plane by
\begin{equation}
\theta^j(u(z)+\vec{A})=\theta(u(z^j)+\vec{A}).
\end{equation}
These functions will have jump discontinuities in the complex
$z$-plane. By using the periodicity of the theta function
(\ref{eq:period}), we can compute their jump discontinuities.
\begin{proposition}\label{pro:thetajump}
Let $\vec{A}=(A_1,\ldots,A_g)^T$ be a $g$ dimensional vector. The
functions $\theta^j(u(z)+\vec{A})$ are analytic in
$\mathbb{C}\setminus(\mathbb{R}\cup S_{\sigma-\mu_2})$. On
$\mathbb{R}\cup S_{\sigma-\mu_2}$, they satisfy the following
conditions
\begin{equation}
\begin{split}
\theta_{\pm}^1(u(z)+\vec{A})&=\theta_{\mp}^2(u(z)+\vec{A}),\quad
z\in\Xi_j,\quad j=1,\ldots, g+1,\\
\theta_+^l(u(z)+\vec{A})&=\theta_-^l(u(z)+\vec{A})e^{(-1)^l2\pi
i\left(u_j(z)+A_j+\frac{\Pi_{jj}}{2}\right)},\quad
z\in\tilde{\Xi}_j,\quad j=1,\ldots, g,\quad l=1,2,\\
\theta_+^l(u(z)+\vec{A})&=\theta_-^l(u(z)+\vec{A}),\quad
z\in\left(\tilde{\Xi}_0\cup\tilde{\Xi}_{g+1}\right),\quad l=1,2,\\
\theta^2_{\pm}(u(z)+\vec{A})&=\theta^3_{\mp}(u(z)+\vec{A}),\quad z\in S_{\sigma-\mu_2},\\
\theta^3_{\pm}(u(z)+\vec{A})&=\theta^4_{\mp}(u(z)+\vec{A}),\quad
z\in\mathbb{R},\\
\theta^l_{+}(u(z)+\vec{A})&=\theta^l_-(u(z)+\vec{A}),\quad z\in
S_{\sigma-\mu_2},\quad l=1,4.
\end{split}
\end{equation}
\end{proposition}
\begin{proof} Let us first consider the discontinuities of
$\theta^1(u(z)+\vec{A})$ and $\theta^2(u(z)+\vec{A})$ across
$\Xi_j$. Let $\pi:\mathcal{\Lie}\rightarrow\overline{\mathbb{C}}$ be
the projection of $\mathcal{\Lie}$ onto the Riemann sphere. Suppose
$z$ is a point in $\Xi_j$. Let $z\in\mathbb{C}$ and define the
points $z_{\pm\epsilon}^j\in\mathcal{\Lie}_j$ to be
\begin{equation}\label{eq:zepsilon}
\pi(z_{\pm i\epsilon}^j)=z\pm i\epsilon,\quad z_{\pm
i\epsilon}^j\in\mathcal{\Lie}_j.
\end{equation}
We will now choose $\epsilon>0$ to be real and positive and let
$z\in\Xi_j$. From the definition of the integration contour in
Figure \ref{fig:contour} and the canonical basis of cycles in Figure
\ref{fig:cycle}, we have the following relation between the points
$z_{\pm i\epsilon}^j$ as $\epsilon\rightarrow 0$.
\begin{equation}\label{eq:ak}
\begin{split}
u(z_{\pm i\epsilon}^1)&=u(z_{\mp i\epsilon}^2),\quad
z\in\Xi_{g+1},\\
u(z_{\pm i\epsilon}^1)&=u(z_{\mp i\epsilon}^2)+\sum_{k=j}^{g}\oint_{a_k}\D\omega,\\
&=u(z_{\mp i\epsilon}^2)+\sum_{k=j}^{g}\vec{e}^k,\quad
z\in\Xi_j,\quad j=1,\ldots, g.
\end{split}
\end{equation}
where $\vec{e}^k$ is a vector with 1 in the $k^{th}$ entry and zero
elsewhere and $\D\omega$ is the vector
\begin{equation}
\D\omega=(\D\omega_1,\ldots,\D\omega_g)^T.
\end{equation}
From this and the periodicity of the theta function
(\ref{eq:period}), we obtain
\begin{equation}
\theta_{\pm}^1(u(z)+\vec{A})=\theta_{\mp}^2(u(z)+\vec{A}),\quad
z\in\Xi_j,\quad j=1,\ldots, g+1.
\end{equation}
Let us now consider a point $z$ in $\tilde{\Xi}_j$. Again, from the
definition of the integration contour and the canonical basis, we
have, as $\epsilon\rightarrow 0$, the following
\begin{equation}
\begin{split}
u(z_{i\epsilon}^l)&=u(z_{-i\epsilon}^l)+(-1)^{l+1}\oint_{b_j}\D\omega,\\
&=u(z_{-i\epsilon}^l)+(-1)^{l+1}\Pi\vec{e}^j, \quad
z\in\tilde{\Xi}_j,\quad j=1,\ldots, g.
\end{split}
\end{equation}
where $l=1,2$. From this and the periodicity of the theta function,
we see that
\begin{equation}
\begin{split}
\theta_+^l(u(z)+\vec{A})&=\theta_-^l(u(z)+\vec{A})e^{(-1)^l2\pi
i\left(u_j(z)+A_j+\frac{\Pi_{jj}}{2}\right)},\quad
z\in\tilde{\Xi}_j,\quad j=1,\ldots, g.
\end{split}
\end{equation}
From the definition of the integration contour, it is clear that
$\theta^1(u(z)+\vec{A})$ and $\theta^2(u(z)+\vec{A})$ are analytic
across $\mathbb{R}\setminus \left(\lambda_1,\lambda_{2g+2}\right)$
and that $\theta^1(u(z)+\vec{A})$ is analytic across
$S_{\sigma-\mu_2}$.

Let us now consider the discontinuities of $\theta^3(u(z)+\vec{A})$
and $\theta^2(u(z)+\vec{A})$ on $S_{\sigma-\mu_2}$. Let $z$ be a
point on $S_{\sigma-\mu_2}$, from the definition of the contours, it
follows immediately that
\begin{equation}
\begin{split}
\theta_{+}^2(u(z)+\vec{A})&=\theta_{-}^3(u(z)+\vec{A}),\quad z\in
S_{\sigma-\mu_2}.
\end{split}
\end{equation}
Let us now consider the boundary value of $\theta^2(u(z)+\vec{A})$
on the minus side of $S_{\sigma-\mu_2}$. For small and positive
$\epsilon\rightarrow0$, we have
\begin{equation}
\begin{split}
u(z^2+\epsilon)+\oint_{\Sigma}\D\omega&=u(z^3-\epsilon),
\end{split}
\end{equation}
where $\Sigma$ is the close loop on $\mathcal{\Lie}$ depicted in
Figure \ref{fig:gamma}. Since this loop is contractible, we have
\begin{equation}
\begin{split}
\theta_{-}^2(u(z)+\vec{A})&=\theta_{+}^3(u(z)+\vec{A}),\quad z\in
S_{\sigma-\mu_2}.
\end{split}
\end{equation}
Finally, the conditions
\begin{equation}
\begin{split}
\theta^3_{\pm}(u(z)+\vec{A})&=\theta^4_{\mp}(u(z)+\vec{A}),\quad
z\in\mathbb{R},\\
\theta^4_{+}(u(z)+\vec{A})&=\theta^4_-(u(z)+\vec{A}),\quad z\in
S_{\sigma-\mu_2}.
\end{split}
\end{equation}
follow directly from the definition of the contour of integration.
\end{proof}
\begin{figure}
\centering \psfrag{lambda1}[][][1][0.0]{\tiny$\lambda_1$}
\psfrag{l2}[][][1][0.0]{\tiny$\lambda_2$}
\psfrag{l3}[][][1][0.0]{\tiny$\lambda_3$}
\psfrag{l4}[][][1][0.0]{\tiny$\lambda_4$}
\psfrag{l2g}[][][1][0.0]{\tiny$\lambda_{2g+2}$}
\psfrag{gamma}[][][1][0.0]{\tiny$\Sigma$}
\psfrag{z}[][][1][0.0]{\tiny$z$} \psfrag{it}[][][1][0.0]{\tiny$ic$}
\psfrag{-it}[][][1][0.0]{\tiny$-ic$}
\psfrag{ia}[][][1][0.0]{\tiny$i\zeta$}
\psfrag{-ia}[][][1][0.0]{\tiny$-i\zeta$}
\includegraphics[scale=0.75]{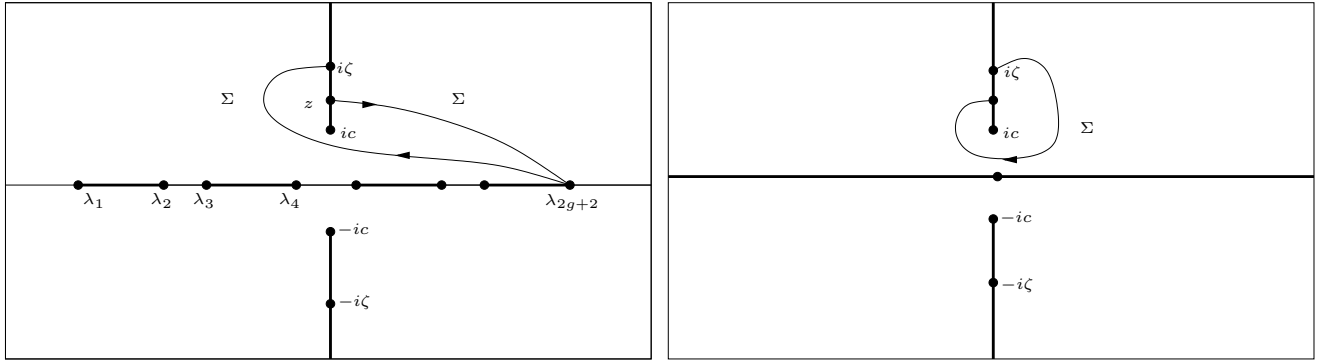}
\caption{The contour $\Sigma$.}\label{fig:gamma}
\end{figure}
\subsection{Meromorphic differentials}
Another key ingredient in the construction of the parametrix is
meromorphic differentials on the Riemann surface. Most of the
results that we will be using can be found in \cite{FK}.

\begin{proposition}\label{pro:mero} Let $d_1,\ldots, d_k$ be $k$ distinct points on a
Riemann surface $\mathcal{\Lie}$. Let $c_1,\ldots, c_k$ be complex
numbers with $\sum_{j=1}^kc_k=0$. Then there exists a meromorphic
1-form $\D\Omega$ on $\mathcal{\Lie}$, holomorphic on
$\mathcal{\Lie}\setminus\{d_1,\ldots,d_k\}$ such that
\begin{equation}
\D\Omega(x)=\left(\frac{c_j}{x_j}+O(1)\right)dx_j,\quad x\rightarrow
d_j,\quad j=1,\ldots, k.
\end{equation}
where $x_j$ is a local coordinate near $d_j$ such that $x_j=0$ at
$d_j$.
\end{proposition}
This result can be found for example, in \cite{FK}. (p.52, theorem
II.5.3)

A meromorphic 1-form with simple poles only is called a meromorphic
1-form of the third type. Let $\D\Omega$ be a meromorphic 1-form of
the third type. In order to define the periods of $\D\Omega$
unambiguously, one has to define the periods to be integrals around
close loops $\hat{a}_j$ and $\hat{b}_j$ that are homologic to the
$a$ and $b$-cycles in Figure \ref{fig:cycle} in
$\mathcal{\Lie}\setminus\Omega_{pole}$, where $\Omega_{pole}$ is the
set of poles of $\D\Omega$.

By adding suitable multiples of holomorphic 1-forms to a given
meromorphic 1-form, we can obtain meromorphic 1-forms with arbitrary
$a$-periods. For example, if a meromorphic 1-form given by
Proposition \ref{pro:mero} has the following $a$-periods
\begin{equation*}
\oint_{a_j}\D\Omega=\mathcal{A}_j,\quad j=1,\ldots,g.
\end{equation*}
Then the meromorphic 1-form
\begin{equation*}
\D\tilde{\Omega}=\D\Omega+\sum_{j=1}^g(\tilde{\mathcal{A}}_j-\mathcal{A}_j)\D\omega_j
\end{equation*}
will have $a$-periods
\begin{equation*}
\oint_{a_j}\D\tilde{\Omega}=\tilde{\mathcal{A}}_j,\quad
j=1,\ldots,g.
\end{equation*}
but with the same pole structure and residues. Of course, we can not
control both the $a$ and the $b$-periods of the 1-form. In fact,
meromorphic 1-forms with prescribed $a$-period and pole structure is
uniquely defined. A meromorphic 1-form with all $a$-periods zero is
called a normalized meromorphic 1-form.
\begin{proposition} Let $d_1,\ldots, d_k$ be $k$ distinct points on a
Riemann surface $\mathcal{\Lie}$. Let $c_1,\ldots, c_k$ be complex
numbers with $\sum_{j=1}^kc_k=0$ and let
$\mathcal{A}_1,\ldots,\mathcal{A}_g$ be arbitrary complex numbers.
Then there exists a unique meromorphic 1-form of the third type
$\D\Omega$ on $\mathcal{\Lie}$, holomorphic on
$\mathcal{\Lie}\setminus\{d_1,\ldots,d_k\}$ such that
\begin{equation}\label{eq:mero}
\begin{split}
\Res_{z=d_j}\D\Omega=c_j,\quad j=1,\ldots, k,\\
\oint_{a_j}\D\Omega=\mathcal{A}_j,\quad j=1,\ldots, g.
\end{split}
\end{equation}
\end{proposition}
\begin{proof} We have already shown the existence part. To see the
uniqueness part, let $\D\Omega$ and $\D\Omega^{\prime}$ be 2
meromorphic 1-forms of the third type with the properties
(\ref{eq:mero}). Let $\D\tilde{\Omega}$ be their difference. Then,
since both $\D\Omega$ and $\D\Omega^{\prime}$ have the same singular
behavior at the points $d_j$, the 1-form $\D\tilde{\Omega}$ does not
have any pole and is therefore holomorphic. Moreover, all its
$a$-periods vanish. Since a holomorphic 1-form with vanishing
$a$-periods has to be zero itself, (See, for example, \cite{FK},
p.65, Proposition III.3.3) the proposition is proven.
\end{proof}
We will conclude this section with a result that relates the periods
of a normalized meromorphic 1-form to the values of the Abel map at
its poles.
\begin{theorem}\label{thm:RB} (See e.g. \cite{FK}, p.65, III.3)
Let $\eta$ be a meromorphic differential of the third type with
simple poles at the points $d_i\in\mathcal{\Lie}$ and $\tilde{\eta}$
be a holomorphic differential. Let $\Pi^i$ and $\tilde{\Pi}^i$ be
their periods
\begin{equation}
\begin{split}
\int_{a_i}\eta&=\Pi^i, \quad \int_{b_i}\eta=\Pi^{i+g}\\
\int_{a_i}\tilde{\eta}&=\tilde{\Pi}^i, \quad
\int_{b_i}\tilde{\eta}=\tilde{\Pi}^{i+g}
\end{split}
\end{equation}
Then the Riemann bilinear relation is the following
\begin{equation}\label{eq:RB}
\sum_{i=1}^g\tilde{\Pi}^i\Pi^{i+g}-\tilde{\Pi}^{g+i}\Pi^{i}=2\pi
i\sum_{d_i}\Res_{d_i}(\eta)\int_{x_0}^{d_i}\tilde{\eta},
\end{equation}
where $x_0$ is an arbitrary point on $\mathcal{\Lie}$.
\end{theorem}

\section{Construction of the outer parametrix}\label{se:outer}
We will now construct the local parametrix with the the theta
function and meromorphic 1-forms.

Let us now define a local coordinate $w$ near $\infty^2$, the point
on $\mathcal{\Lie}_2$, $\mathcal{\Lie}_3$ and $\mathcal{\Lie}_4$
that projects onto $\infty$ in the Riemann sphere.
\begin{equation}\label{eq:w}
\begin{split}
w&=\left\{
    \begin{array}{ll}
      z^{-\frac{1}{3}}, & \hbox{in the first and fourth quadrants of $\mathcal{\Lie}_2$;} \\
      \omega^{2} z^{-\frac{1}{3}}, & \hbox{in the second quadrant of $\mathcal{\Lie}_2$;} \\
      \omega z^{-\frac{1}{3}}, & \hbox{in the third quadrant of
$\mathcal{\Lie}_2$.}
    \end{array}
  \right.\\
w&=\left\{
     \begin{array}{ll}
       \omega^{2} z^{-\frac{1}{3}}, & \hbox{in the first quadrant of $\mathcal{\Lie}_3$;} \\
       z^{-\frac{1}{3}}, & \hbox{in the second and third quadrants of $\mathcal{\Lie}_3$;} \\
       \omega z^{-\frac{1}{3}}, & \hbox{in the fourth quadrant of $\mathcal{\Lie}_3$.}
     \end{array}
   \right.\\
w&=\left\{
     \begin{array}{ll}
       \omega z^{-\frac{1}{3}}, & \hbox{in the first and second quadrants of $\mathcal{\Lie}_4$;} \\
       \omega^{2} z^{-\frac{1}{3}}, & \hbox{in the third and fourth quadrants of $\mathcal{\Lie}_4$.}
     \end{array}
   \right.
\end{split}
\end{equation}
where $\omega=e^{\frac{2\pi i}{3}}$ and the branch of
$z^{\frac{1}{3}}$ is chosen such that $\arg z\in(-\pi,\pi)$. One can
check that $w$ is indeed holomorphic in $\mathcal{\Lie}$ in a
neighborhood of $\infty^2$.

Let us now define four meromorphic 1-forms of the third type
$\D\Delta_j$, $j=1,\ldots,4$ by the following properties.
\begin{definition}\label{de:deltaj}
The normalized meromorphic 1-forms $\D\Delta_j$ are holomorphic in
\begin{equation*}
\mathcal{\Lie}\setminus\{\pm it,
\lambda_{1}^1,\ldots,\lambda_{2g+2}^1,\iota_1,\ldots,\iota_g,\infty^1,\infty^2\}.
\end{equation*}
where $\pm it$ are the points in $\mathcal{\Lie}_2$ that project
onto $\pm ic$ and $\iota_k$ are the zeros of $\theta(u(x))$. At
these points they have simple poles with residues
\begin{equation}\label{eq:deltaj}
\begin{split}
\Res_{\lambda_k^1}\D\Delta_j&=\Res_{\pm
it}\D\Delta_j=-\frac{1}{2},\quad k=1,\ldots,2g+2,\\
\Res_{\iota_k}\D\Delta_j&=1,\quad k=1,\ldots,g,\\
\Res_{\infty^1}\D\Delta_1&=0,\quad \Res_{\infty^2}\D\Delta_1=2,\\
\Res_{\infty^1}\D\Delta_2&=3,\quad \Res_{\infty^2}\D\Delta_2=-1,\\
\Res_{\infty^1}\D\Delta_3&=2,\quad \Res_{\infty^2}\D\Delta_1=0,\\
\Res_{\infty^1}\D\Delta_4&=1,\quad \Res_{\infty^2}\D\Delta_4=1,
\end{split}
\end{equation}
provided none of the $\iota_l$ is equal to $\lambda_{k}^1$ for some
$k$. If some $\iota_l$ is equal to $\lambda_k^1$ for some $k$, then
the residue at $\iota_l$ will be $\frac{1}{2}$.

These 1-forms are then uniquely defined. We will denote the
$b$-period of these 1-forms by $\beta_j$.
\begin{equation}
\beta_j=\left(\oint_{b_1}\D\Delta_j,\ldots,\oint_{b_g}\D\Delta_j\right)^T,\quad
j=1,\ldots, 4.
\end{equation}
To avoid ambiguity in the $b$-periods, let $\pi(\iota_k)$ be the
projection of $\iota_k$ on $\tilde{\Xi}_k$ (See remark
\ref{re:iota}). Then the $b$-periods are computed as integrals on
$b$-cycles in $\mathcal{\Lie}_1$ that intersects $\tilde{\Xi}_k$ at
any point $x<\pi(\iota_k)$ if $\pi(\iota_k)\neq\lambda_{2k}$. If
$\pi(\iota_k)=\lambda_{2k}$, then the $b$-cycle can intersect
$\tilde{\Xi}_k$ at any point $x\neq\lambda_{2k}$ in
$\mathcal{\Lie}_1$.
\end{definition}
We will now define four functions in the Riemann surface
$\mathcal{\Lie}$. First let $\Xi_{k}^{\pm}\in\mathcal{\Lie}$ be the
images of $\Xi_k$ under the maps $\xi_{1,\pm}(z)$, that is,
\begin{equation}\label{eq:xipm}
\begin{split}
\Xi_k^{\pm}=\left\{(z,\xi)|z\in\Xi_k,\quad
\xi=\xi_{1,\pm}(z)\right\},\quad k=1,\ldots,g+1.
\end{split}
\end{equation}
Let $z_0$ be a point in $\Xi_{g+1}^-$. The exact choice of $z_0$ is
immaterial to the construction as long as $z_0\neq\lambda_{2g+1}^1$
or $\lambda_{2g+2}^1$. We will now define the functions $N_j(z)$ on
$\mathcal{\Lie}$ as follows.
\begin{equation}\label{eq:Nj}
\begin{split}
N_j(z)&=e^{\Delta_j(z)}\frac{\theta\left(u(z)+\frac{\beta_j}{2\pi
i}+n\vec{\alpha}\right)}
{\theta\left(u(z)\right)}\\
&=e^{\Delta_j(z)}\Theta_j(z),\quad
\vec{\alpha}=\left(\alpha_1,\ldots,\alpha_{g}\right)^T,\quad
j=1,\ldots,4.
\end{split}
\end{equation}
where the function $\Delta_j(z)$ is given by
$\Delta_j(z)=\int_{z_0}^z\D\Delta_j$ and the path of integration is
defined in the same way as the ones for the Abel map, except that
every path now starts at $z_0$.

Let $z^k$ be the point on $\mathcal{\Lie}_k$ that projects to $z$ in
$\overline{\mathbb{C}}$. As before, we will now define four
functions $e^{\Delta_j^k(z)}$ on the complex $z$-plane by
\begin{equation}
e^{\Delta_j^k(z)}=e^{\Delta_j(z^k)}.
\end{equation}
Then these functions have the following jump discontinuities in the
complex $z$-plane.
\begin{proposition}\label{pro:deltajump}
The functions $e^{\Delta_j^l(z)}$ are analytic in
$\mathbb{C}\setminus(\mathbb{R}\cup S_{\sigma-\mu_2})$. On
$\mathbb{R}\cup S_{\sigma-\mu_2}$, they satisfy the following
conditions
\begin{equation}\label{eq:deltajump}
\begin{split}
e^{\Delta_{j,\pm}^1(z)}&=\mp e^{\Delta_{j,\mp}^2(z)},\quad
z\in\Xi_k,\quad k=1,\ldots, g+1,\\
e^{\Delta_{j,+}^l(z)}&=e^{\Delta_{j,-}^l(z)+(-1)^{l-1}\left(\beta_j\right)_k},\quad
z\in\tilde{\Xi}_k,\quad k=1,\ldots, g,\quad l=1,2,\\
e^{\Delta_{j,+}^l(z)}&=e^{\Delta_{j,-}^l(z)},\quad
z\in\left(\tilde{\Xi}_0\cup\tilde{\Xi}_{g+1}\right),\quad l=1,2,\\
e^{\Delta_{j,\pm}^2(z)}&=\pm e^{\Delta_{j,\mp}^3(z)},\quad z\in S_{\sigma-\mu_2},\\
e^{\Delta_{j,\pm}^3(z)}&=e^{\Delta_{j,\mp}^4(z)},\quad
z\in\mathbb{R},\\
e^{\Delta_{j,+}^l(z)}&=e^{\Delta_{j,-}^l(z)},\quad z\in
S_{\sigma-\mu_2},\quad l=1,4.
\end{split}
\end{equation}
where $\left(\beta_j\right)_k$ is the $k^{th}$ component of the
vector $\beta_j$.
\end{proposition}
\begin{proof} The proof follows similar argument as the ones used in
the proof of Proposition \ref{pro:thetajump}. First let us consider
the jump discontinuities on $\Xi_k$. Let $z$ be a point in $\Xi_k$
and define the points $z_{\pm i\epsilon}^l$ as in
(\ref{eq:zepsilon}) in the proof of Proposition \ref{pro:thetajump}.
First consider the boundary values $e^{\Delta_{j,+}^1(z)}$ and
$e^{\Delta_{j,-}^2(z)}$. Choose integration contours $\Sigma_+$ and
$\Sigma_-$ from $z_0$ to the points $z_{i\epsilon}^1$ and
$z_{-i\epsilon}^2$ as in Figure \ref{fig:gammadelta}. Let
$\Sigma=\Sigma_+-\Sigma_-$. Then $\Sigma$ can be deformed into the
sum $\sum_{l=k}^ga_l$ of the $a$-cycles and a loop $\Sigma_{2g+2}$
around the point $\lambda_{2g+2}^1$ in
$\mathcal{\Lie}\setminus\Delta_{pole}$, where $\Delta_{pole}$ is the
set of poles of $\D\Delta_j$. (See Figure \ref{fig:loop}).
\begin{equation}\label{eq:deltapole}
\Delta_{pole}=\{\pm it,
\lambda_{1}^1,\ldots,\lambda_{2g+2}^1,\iota_1,\ldots,\iota_g,\infty^1,\infty^2\}.
\end{equation}
\begin{figure}
\centering \psfrag{lambda1}[][][1][0.0]{\tiny$\lambda_1$}
\psfrag{l2}[][][1][0.0]{\tiny$\lambda_2$}
\psfrag{l3}[][][1][0.0]{\tiny$\lambda_3$}
\psfrag{l4}[][][1][0.0]{\tiny$\lambda_4$}
\psfrag{l2g}[][][1][0.0]{\tiny$\lambda_{2g+2}$}
\psfrag{g1}[][][1][0.0]{\tiny$\Sigma_+$}
\psfrag{g2}[][][1][0.0]{\tiny$\Sigma_-$}
\psfrag{it}[][][1][0.0]{\tiny$ic$}
\psfrag{-it}[][][1][0.0]{\tiny$-ic$}
\includegraphics[scale=1]{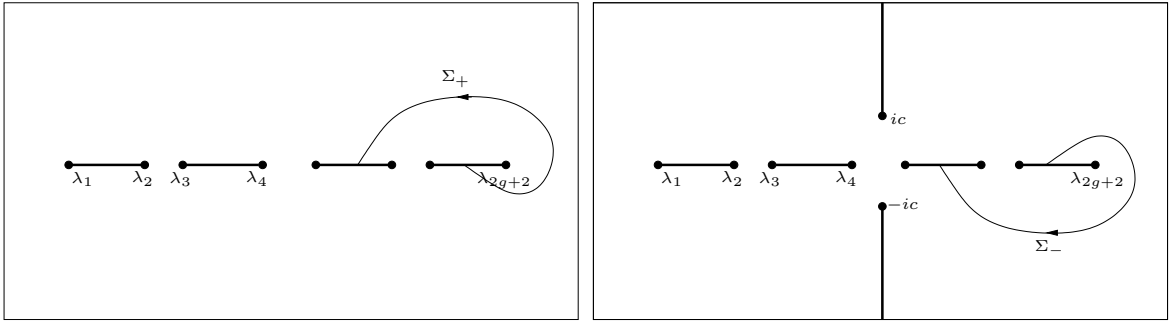}
\caption{The contours $\Sigma_{\pm}$ for $e^{\Delta_{j,+}^1(z)}$ and
$e^{\Delta_{j,-}^2(z)}$.}\label{fig:gammadelta}
\end{figure}
\begin{figure}
\centering \psfrag{lambda1}[][][1][0.0]{\tiny$\lambda_1$}
\psfrag{l2}[][][1][0.0]{\tiny$\lambda_2$}
\psfrag{l3}[][][1][0.0]{\tiny$\lambda_3$}
\psfrag{l4}[][][1][0.0]{\tiny$\lambda_4$}
\psfrag{l2g}[][][1][0.0]{\tiny$\lambda_{2g+2}$}
\psfrag{g1}[][][1][0.0]{\tiny$\Sigma$}
\psfrag{g2}[][][1][0.0]{\tiny$\Sigma_{2g+2}$}
\psfrag{it}[][][1][0.0]{\tiny$ic$}
\psfrag{-it}[][][1][0.0]{\tiny$-ic$}
\includegraphics[scale=1]{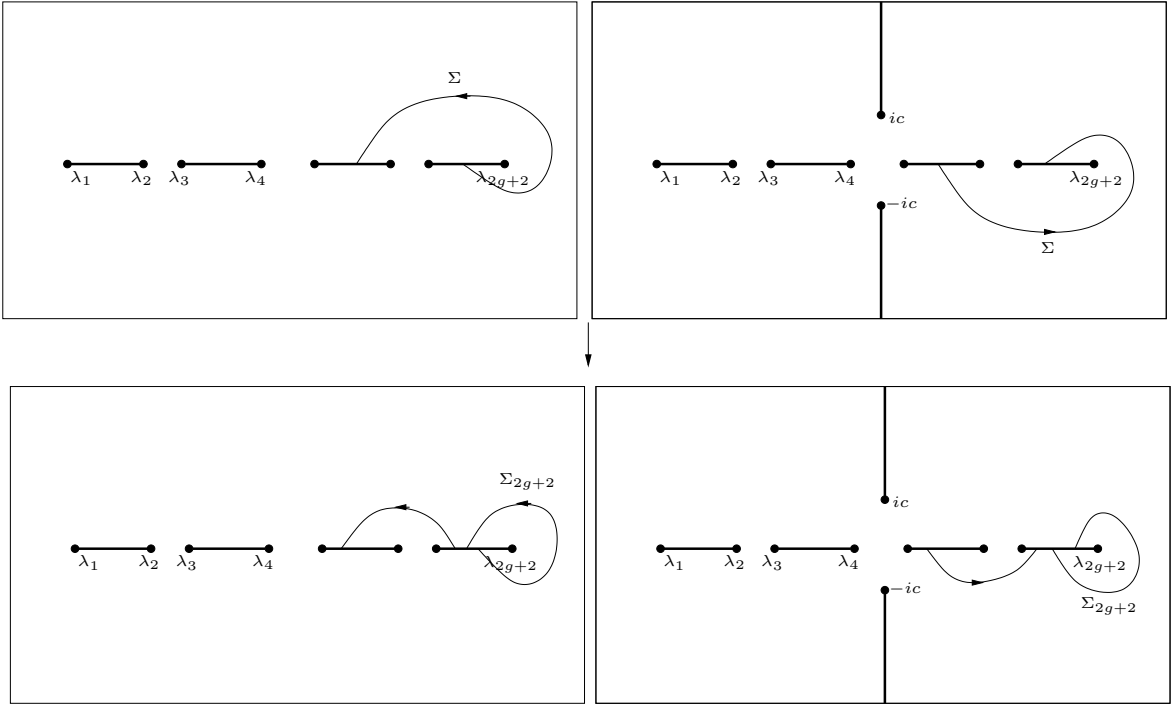}
\caption{The contours deformation of the loop $\Sigma$ for
$e^{\Delta_{j,+}^1(z)}$ and
$e^{\Delta_{j,-}^2(z)}$.}\label{fig:loop}
\end{figure}
Therefore we have
\begin{equation}\label{eq:akdel}
\begin{split}
\exp\left(\Delta_j(z_{i\epsilon}^1)\right)&=\exp\left(\Delta_j(z_{-i\epsilon}^2)+\sum_{l=k}^{g}\oint_{a_l}\D\Delta_j
+\oint_{\Sigma_{2g+2}}\D\Delta_j\right),\\
&=-\exp\left(\Delta_j(z_{ -i\epsilon}^2)\right),\quad
z\in\Xi_k,\quad k=1,\ldots, g.
\end{split}
\end{equation}
where the last equality follows from the fact that $\D\Delta_j$ has
residue $-\frac{1}{2}$ at the point $\lambda_{2g+2}^1$.

Let us now consider the boundary values $e^{\Delta_{j,-}^1(z)}$ and
$e^{\Delta_{j,+}^2(z)}$ on $\Xi_k$. Since $z_0\in\Xi_{g+1}^-$, the
integration contours $\Sigma_-$ and $\Sigma_+$ can now be chosen to
lie in the lower (upper) half plane of $\mathcal{\Lie}_1$
($\mathcal{\Lie}_2$). The loop $\Sigma=\Sigma_+-\Sigma_-$ can now be
deformed into the sum $-\sum_{l=k}^ga_l$ of the $a$-cycles. However,
such a deformation will necessarily go pass the poles
$\lambda_{2k}^1,\ldots,\lambda_{2g+1}^1$ and
$\iota_k,\ldots,\iota_{g}$ of $\D\Delta_j$ (Recall that by the
remark after Proposition \ref{pro:constant}, there is exactly one
point $\iota_k$ that belongs to
$\tilde{\Xi}_k^1\cup\tilde{\Xi}_k^2$). Since the residues of
$\D\Delta_j$ at these points are given by $-\frac{2(g-k)}{2}$ from
the $\lambda_l^1$ and $g-k$ from the $\iota_l$ when all $\iota_l$
and $\lambda_m^1$ are distinct, the total residue at these points is
zero. It is clear from Definition \ref{de:deltaj} that, when some
$\iota_l$ coincide with the $\lambda_m^1$, the total residue at
these points remains unchanged. Hence we have
\begin{equation*}
\begin{split}
\exp\left(\Delta_j(z_{-i\epsilon}^1)\right)&=\exp\left(\Delta_j(z_{i\epsilon}^2)-\oint_{\Sigma}\D\Delta_j\right),\\
&=\exp\left(\Delta_j(z_{i\epsilon}^2)\right),\quad z\in\Xi_k,\quad
k=1,\ldots, g.
\end{split}
\end{equation*}
Let us now consider the boundary values on the gaps $\tilde{\Xi}_k$.
Let $z\in\tilde{\Xi}_k$. For the boundary values
$e^{\Delta_{j,\pm}^1(z)}$, we choose $\Sigma_{\pm}$ to be
integration contours that go from $z_0$ to $z_{\pm i\epsilon}^1$ in
$\mathcal{\Lie}_1$ without intersecting $(-\infty,\lambda_{2g+2}^1)$
except at $z_0$ and $z_{\pm i\epsilon}^1$. Let the loop $\Sigma$ be
$\Sigma=\Sigma_+-\Sigma_-$, then for $k=1,\ldots,g$, $\Sigma$ can be
deformed into the $b$-cycle $b_k$ without passing any pole of
$\D\Delta_j$, except possibly $\iota_k$ (Recall the definition of
the $b$-periods of $\D\Delta_j$ in Definition \ref{de:deltaj}). When
$\iota_k$ is not equal to $\lambda_{2k}^1$ or $\lambda_{2k+1}^1$,
$\D\Delta_j$ has integer residue at $\iota_k$, and when $\iota_k$ is
equal to either $\lambda_{2k}^1$ or $\lambda_{2k+1}^1$, the
deformation from $\Sigma$ to $b_k$ will not have to go pass
$\iota_k$. This implies
\begin{equation*}
e^{\Delta_{j,+}^1(z)}=e^{\Delta_{j,-}^1(z)+\left(\beta_j\right)_k},\quad
z\in\tilde{\Xi}_k,\quad k=1,\ldots, g.
\end{equation*}
For $k=0$, the loop $\Sigma$ can be deformed into a small loop
around the point $\infty^1$. Since the 1-forms $\D\Delta_j$ have
integer residues at $\infty^1$, we have
\begin{equation*}
e^{\Delta_{j,+}^1(z)}=e^{\Delta_{j,-}^1(z)},\quad z\in\tilde{\Xi}_0.
\end{equation*}
If $k=g+1$, then the loop $\Sigma$ will be contractible in
$\mathcal{\Lie}\setminus\Delta_{pole}$. Hence we have
\begin{equation*}
e^{\Delta_{j,+}^1(z)}=e^{\Delta_{j,-}^1(z)},\quad
z\in\tilde{\Xi}_{g+1}.
\end{equation*}
On the other hand, for the boundary values $e^{\Delta_{j,\pm}^2(z)}$
on $\tilde{\Xi}_k$, let us consider $\Sigma_{\pm}$ to be integration
contours that go from $z_0$ to $z_{\pm i\epsilon}^2$ in
$\mathcal{\Lie}_2$ without intersecting $(-\infty,\lambda_{2g+2}^1)$
except at $z_0$ and $z_{\pm i\epsilon}^2$. Let $k=1,\ldots, g$, and
let $\pi(\iota_k)$ be the projection of $\iota_k$ onto the
$z$-plane. Then depending on the relative positions of $z$,
$\pi(\iota_k)$ and $\lambda_{2k}^1$, the loop
$\Sigma=\Sigma_+-\Sigma_-$ can be deformed into $-b_k$, together
with small loops around the poles $\iota_k,\ldots,\iota_g$ and
$\lambda_{2k+1}^1,\ldots,\lambda_{2g+2}^1$ in
$\mathcal{\Lie}\setminus\Delta_{pole}$; or it can be deformed into
the sum of $-b_k$ and small loops around the poles
$\iota_{k+1},\ldots,\iota_g$ and
$\lambda_{2k+1}^1,\ldots,\lambda_{2g+2}^1$ in
$\mathcal{\Lie}\setminus\Delta_{pole}$. In either cases, the total
residue of $\D\Delta_j$ at these points will be an integer.
Therefore we have
\begin{equation*}
e^{\Delta_{j,+}^2(z)}=e^{\Delta_{j,-}^2(z)-\left(\beta_j\right)_k},\quad
z\in\tilde{\Xi}_k,\quad k=1,\ldots, g.
\end{equation*}
Similarly, for $k=0$, the loop $\Sigma$ can be deformed into a small
loop around the points $\infty^2$ and $\pm it$. Since the total
residue the 1-form $\D\Delta_j$ at these points is an integer, we
have
\begin{equation*}
e^{\Delta_{j,+}^2(z)}=e^{\Delta_{j,-}^2(z)},\quad z\in\tilde{\Xi}_0.
\end{equation*}
If $k=g+1$, then the loop $\Sigma$ will be contractible in
$\mathcal{\Lie}\setminus\Delta_{pole}$. Hence we have
\begin{equation*}
e^{\Delta_{j,+}^2(z)}=e^{\Delta_{j,-}^2(z)},\quad
z\in\tilde{\Xi}_{g+1}.
\end{equation*}
We now consider the boundary values $e^{\Delta_{j,-}^2(z)}$ and
$e^{\Delta_{j,+}^3(z)}$ at $S_{\sigma-\mu_2}$, let $z\in
S_{\sigma-\mu_2}$. Let us again denote by $\Sigma_+$ and $\Sigma_-$
contours of integration from $z_0$ to $z-\epsilon$ in
$\mathcal{\Lie}_3$ and $z+\epsilon$ in $\mathcal{\Lie}_2$. Then
depending on whether $z$ is in the upper or lower half plane, the
loop $\Sigma=\Sigma_+-\Sigma_-$ can be deformed into to a small loop
around the pole $it$ or $-it$ in
$\mathcal{\Lie}\setminus\Delta_{pole}$ (See Figure \ref{fig:gamma}.
The loop $\Sigma$ in this case is the same except that it begins and
ends at $z_0$ instead of $\lambda_{2g+2}^1$). Since the residue of
$\D\Delta_j$ around $it$ or $-it$ is $-\frac{1}{2}$, we have
\begin{equation}
e^{\Delta_{j,+}^3(z)}=-e^{\Delta_{j,-}^2(z)},\quad z\in
S_{\sigma-\mu_2}.
\end{equation}
The rest of the jump discontinuities in (\ref{eq:deltajump}) now
follow directly from the definition of the integration contours as
in the proof of Proposition \ref{pro:thetajump}.
\end{proof}

Let us denote by $N_j^k(z)$ the projection of $N_j(z)$ onto the
$k^{th}$-sheet, that is,
$N_j^k(z)=N_j(z^k)=N_j\left(z,\xi_k(z)\right)$, where $\xi_k(z)$ is
the function $\xi(z)$ on $\mathcal{\Lie}_k$. Then we have the
following.
\begin{theorem}\label{thm:outer}
Let $N(z)$ be the $4\times4$ matrix whose elements are given by
\begin{equation}\label{eq:nmatrix}
N_{jk}(z)=\left\{
            \begin{array}{ll}
              N_j^k(z), & \hbox{$\mathrm{Im} z>0$;} \\
              (-1)^{\delta_{4,k}}N_j^k(z), & \hbox{$\mathrm{Im}z<0$,}
            \end{array}
          \right.
\end{equation}
where $N_j(z)$ are defined in (\ref{eq:Nj}). Suppose we have
\begin{equation}
\theta\left(u(\infty^1)+\frac{\beta_j}{2\pi i}+n\vec{\alpha}\right)
\theta\left(u(\infty^2)+\frac{\beta_j}{2\pi
i}+n\vec{\alpha}\right)\neq 0.
\end{equation}
Let the constants $L_j$ be
\begin{equation}\label{eq:Lj}
\begin{split}
L_1&=N_{1}^{-1}(\infty^1),\quad L_2=\lim_{w\rightarrow
0}N_2^{-1}(w)w^{-1},\\L_3&=\lim_{w\rightarrow 0}N_3^{-1}(w), \quad
L_4=\lim_{w\rightarrow0}N_4^{-1}(w)w,
\end{split}
\end{equation}
where $w$ is the local coordinate near $\infty^2$ defined in
(\ref{eq:w}) and the limits in $L_2$, $L_3$ and $L_4$ are taken as
$z\rightarrow\infty^2$ in the first quadrant of $\mathcal{\Lie}_2$.
Then the matrix
\begin{equation}\label{eq:sol}
S_{\infty}(z)=\begin{pmatrix}1&0&0&0\\0&-\frac{i}{\sqrt{3}}&0&-\frac{i\kappa}{\sqrt{3}}\\0&0&-\frac{i}{\sqrt{3}}&0\\0&0&0&-\frac{i}{\sqrt{3}}\end{pmatrix}
\diag\left(L_1,L_2,L_3,L_4\right)N(z)
\end{equation}
satisfies the Riemann-Hilbert problem (\ref{eq:outer}), where
$\kappa$ is the following in the expansion of $N_2(z)$ at
$z=\infty^2$
\begin{equation*}
N_2(z)=L_2^{-1}\left(w^{-1}+\kappa_{2,0}-\kappa w+O(w^{2})\right),
\end{equation*}
as $z\rightarrow\infty^2$ in the first quadrant of
$\mathcal{\Lie}_2$.
\end{theorem}
\begin{remark}
The constants $L_j$, $j=1,\ldots,4$ can be represented as
\begin{equation}\label{eq:Lexp}
\begin{split}
L_1&=e^{-\Delta_1(\infty^1)}\frac{\theta\left(u(\infty^1)\right)}
{\theta\left(u(\infty^1)+\frac{\beta_1}{2\pi i}+n\vec{\alpha}\right)
},\\
L_2&=\left(\lim_{w\rightarrow
0}e^{-\Delta_2(w)}w^{-1}\right)\frac{\theta\left(u(\infty^2)\right)}
{\theta\left(u(\infty^2)+\frac{\beta_2}{2\pi i}+n\vec{\alpha}\right)
},\\
L_3&=\lim_{w\rightarrow
0}e^{-\Delta_3(w)}\frac{\theta\left(u(\infty^2)\right)}
{\theta\left(u(\infty^2)+\frac{\beta_3}{2\pi i}+n\vec{\alpha}\right)
},\\
L_4&=\left(\lim_{w\rightarrow
0}e^{-\Delta_4(w)}w\right)\frac{\theta\left(u(\infty^2)\right)}
{\theta\left(u(\infty^2)+\frac{\beta_4}{2\pi i}+n\vec{\alpha}\right)
}.
\end{split}
\end{equation}
where the limits are taken as $z\rightarrow\infty^2$ in the first
quadrant of $\mathcal{\Lie}_2$.
\end{remark}
\begin{proof} First note
that, by using Proposition \ref{pro:thetajump} and Proposition
\ref{pro:deltajump}, one can verify that $N(z)$ does satisfy the
jump discontinuities in (\ref{eq:outer}).

Since $N(z)$ satisfies the jump discontinuities of (\ref{eq:outer}),
the matrix $M(z)N^{-1}(z)$ does not have any jump discontinuities in
$\mathbb{C}$. Moreover, this matrix does not grow faster than
$z^{\frac{2}{3}}$ at $z=\infty$ and has at worst square root
singularities at the points $\lambda_j$ and $\pm it$. Since it has
no jump discontinuities, all these singularities are removable and
therefore we have $M(z)=HN(z)$ for some constant matrix $H$. To
determine the constant matrix $H$, we will have to study the
behavior of $N(z)$ as $z\rightarrow\infty$.

The behavior of $M(z)$ is given by the following
\begin{equation}
M(z)=\begin{pmatrix}1+O(z^{-1})
&O(z^{-\frac{2}{3}})&O(z^{-\frac{2}{3}})&O(z^{-\frac{2}{3}})\\
O(z^{-1})&*&*&*\\
O(z^{-1})&*&*&*\\
O(z^{-1})&*&*&*
\end{pmatrix}
\end{equation}
where the $3\times 3$ lower right block is given by
\begin{equation}
\begin{pmatrix}-\frac{i}{\sqrt{3}}z^{\frac{1}{3}}(1+O(z^{-1}))&\frac{\omega
i}{\sqrt{3}}z^{\frac{1}{3}}(1+O(z^{-1}))&\frac{\omega^2i}{\sqrt{3}}z^{\frac{1}{3}}(1+O(z^{-1}))\\
-\frac{i}{\sqrt{3}}(1+O(z^{-\frac{2}{3}}))&\frac{
i}{\sqrt{3}}(1+O(z^{-\frac{2}{3}}))&\frac{i}{\sqrt{3}}(1+O(z^{-\frac{2}{3}}))\\
-\frac{i}{\sqrt{3}}z^{-\frac{1}{3}}(1+O(z^{-\frac{1}{3}}))&\frac{\omega^2
i}{\sqrt{3}}z^{-\frac{1}{3}}(1+O(z^{-\frac{1}{3}}))&\frac{\omega
i}{\sqrt{3}}z^{-\frac{1}{3}}(1+O(z^{-\frac{1}{3}}))
\end{pmatrix}
\end{equation}
for $z\rightarrow\infty$ in the first quadrant. From the relation
between the local coordinate $w$ and $z$ in (\ref{eq:w}) and the
jump discontinuities of $N(z)$ near $\infty^2$, we see that, if we
can show that the functions $N_j(z)$ behave as
\begin{equation}\label{eq:Nbeha}
\begin{split}
N_1(z)&=L_1^{-1}(1+O(z^{-1})),\quad z\rightarrow\infty^1,\quad
N_1(z)=O(w^2),\quad z\rightarrow\infty^2,\\
N_2(z)&=O(z^{-1}),\quad z\rightarrow\infty^1,\quad
N_2(z)=L_2^{-1}\left(w^{-1}-\kappa w+O(w^{2})\right),\quad z\rightarrow\infty^2,\\
N_3(z)&=O(z^{-1}),\quad z\rightarrow\infty^1,\quad
N_3(z)=L_3^{-1}(1+O(w^{2})),\quad z\rightarrow\infty^2,\\
N_4(z)&=O(z^{-1}),\quad z\rightarrow\infty^1,\quad
N_4(z)=L_4^{-1}w(1+O(w)),\quad z\rightarrow\infty^2,
\end{split}
\end{equation}
when $z\rightarrow\infty^2$ in the first quadrant of
$\mathcal{\Lie}_2$, then the matrix in (\ref{eq:sol}) will be the
unique solution of the Riemann-Hilbert problem (\ref{eq:outer}). The
asymptotic behavior of $N_1(z)$ and $N_4(z)$ follows immediately
from the definition of the functions $N_j(z)$ (\ref{eq:Nj}), the
constants $L_j$ (\ref{eq:Lj}) and behavior of the 1-forms
$\D\Delta_j$ (\ref{eq:deltaj}).

We will now prove the equations in (\ref{eq:Nbeha}) for $N_2(z)$ and
$N_3(z)$. Let the involution $\varrho$ on $\mathcal{\Lie}$ be
$\varrho(z,\xi(z))=(-z,\xi(-z))$. To simplify the notation, we shall
simply denote $\varrho(z,\xi(z))$ by $-z$. Let us consider the
functions $N_j(-z)$ for $j=2,3$. The singularity structure of this
function is the same as $N_j(z)$. By Proposition \ref{pro:thetajump}
and \ref{pro:deltajump} and the expression of $N_j(z)$
(\ref{eq:Nj}), we see that the functions $N_2(z)$ and $N_3(z)$
satisfies the following jump discontinuities on $\mathcal{\Lie}$.
\begin{equation}\label{eq:jump23}
\begin{split}
N_{j,+}(z)&=-N_{j,-}(z),\quad z\in \Xi_k^+,\quad k=1,\ldots, g+1,\\
N_{j,+}(z)&=e^{(-1)^l2\pi i n\alpha_k}N_{j,-}(z),\quad z\in \tilde{\Xi}_k^l,\quad k=1,\ldots, g,\quad l=1,2,\\
N_{j,+}(z)&=-N_{j,-}(z),\quad z\in S_{\sigma-\mu_2}^+.
\end{split}
\end{equation}
where $\Xi_k^{\pm}$ is defined in (\ref{eq:xipm}), $\tilde{\Xi}_k^l$
is the interval on $\mathcal{\Lie}_l$ that projects to
$\tilde{\Xi}_k$. That is
\begin{equation*}
\begin{split}
\tilde{\Xi}_k^l=\left\{(z,\xi)|z\in\tilde{\Xi}_k,\quad
\xi=\xi_l(z)\right\}.
\end{split}
\end{equation*}
The intervals $S_{\sigma-\mu_2}^{\pm}$ are defined to be
\begin{equation*}
\begin{split}
S_{\sigma-\mu_2}^{\pm}=\left\{(z,\xi)|z\in S_{\sigma-\mu_2},\quad
\xi=\xi_{3,\pm}(z)\right\}.
\end{split}
\end{equation*}
On the other hand, from (\ref{eq:jump23}), we see that the function
$N_j(-z)$ has the following jump discontinuities
\begin{equation}\label{eq:-jump23}
\begin{split}
N_{j,+}(-z)&=-N_{j,-}(-z),\quad z\in \Xi_k^-,\quad k=1,\ldots, g+1,\\
N_{j,+}(-z)&=e^{(-1)^{l+1}2\pi i n\alpha_{g+1-k}}N_{j,-}(-z),\quad z\in \tilde{\Xi}_k^l,\quad k=1,\ldots, g,\quad l=1,2,\\
N_{j,+}(-z)&=-N_{j,-}(-z),\quad z\in S_{\sigma-\mu_2}^-.
\end{split}
\end{equation}
Note that the union of the contours
$\left(\cup_{k=1}^{g+1}\Xi_k^+\right)$ and
$\left(\cup_{k=1}^{g+1}\Xi_k^-\right)$ divides $\mathcal{\Lie}$ into
2 disjoint sets, which are the first sheet $\mathcal{\Lie}_1$ and
the union of the other sheets
$\mathcal{\Lie}_2\cup\mathcal{\Lie}_3\cup\mathcal{\Lie}_4$.
Similarly, the contour $S_{\sigma-\mu_2}^+\cup S_{\sigma-\mu_2}^-$
divides $\mathcal{\Lie}$ into the sets
$\mathcal{\Lie}_1\cup\mathcal{\Lie}_2$ and
$\mathcal{\Lie}_3\cup\mathcal{\Lie}_4$. Let $\hat{N}_j(z)$ be
\begin{equation}\label{eq:hatN}
\hat{N}_j(z)=\left\{
              \begin{array}{ll}
                N_j(-z), & \hbox{$z\in\mathcal{\Lie}_1\cup\mathcal{\Lie}_3\cup\mathcal{\Lie}_4$;} \\
                -N_j(-z), & \hbox{$z\in\mathcal{\Lie}_2$.}
              \end{array}
            \right.
\end{equation}
Since the constants $\alpha_k$ satisfy the symmetry
$\alpha_k=1-\alpha_{g+1-k}$ (\ref{eq:symmetry}), from
(\ref{eq:jump23}), (\ref{eq:-jump23}) and (\ref{eq:hatN}) we see
that the function
\begin{equation}\label{eq:tildeN}
\tilde{N}_j(z)=\frac{\hat{N}_j(z)}{N_j(z)}
\end{equation}
is either a meromorphic function on $\mathcal{\Lie}$ with poles
exactly at the $g$ zeros of
\begin{equation*}
\theta\left(u(z)+\frac{\beta_j}{2\pi i}+n\vec{\alpha}\right)
\end{equation*}
or it is a constant. By the assumption of the theorem, this theta
function is not identically zero. Hence by Theorem \ref{thm:RR}, we
see that $\tilde{N}_j(z)$ must be a constant $\mathcal{K}_j$. By
using the jump discontinuities (\ref{eq:jump23}), (\ref{eq:-jump23})
of the $N_j(z)$ near $z=\infty$, and the relation between the
coordinate $z^{\frac{1}{3}}$ and $w$, we have the following behavior
of $N_j(z)$ near $\infty^2$
\begin{equation}\label{eq:Ninft}
\begin{split}
N_2(z)&=L_2^{-1}\left(z^{\frac{1}{3}}+\kappa_{2,0}-\kappa
z^{-\frac{1}{3}}+O(w^2)\right),\\
N_3(z)&=L_3^{-1}\left(1+\kappa_{3,0}z^{-\frac{1}{3}}+O(w^2)\right),
\end{split}
\end{equation}
as $z\rightarrow\infty^2$ in the first quadrant of
$\mathcal{\Lie}_2$, where the branch cut of $z^{\frac{1}{3}}$ is
chosen to be the negative real axis. On the other hand, since
$-z\rightarrow\infty^2$ in the third quadrant when
$z\rightarrow\infty^2$ in the first quadrant, the functions
$\hat{N}_j(z)$ have the following behavior
\begin{equation}\label{eq:Nhatinf}
\begin{split}
\hat{N}_2(z)&=L_2^{-1}\left(-z^{\frac{1}{3}}+\kappa_{2,0}+\kappa
z^{-\frac{1}{3}}+O(w^2)\right),\\
\hat{N}_3(z)&=L_3^{-1}\left(1-\kappa_{3,0}z^{-\frac{1}{3}}+O(w^2)\right),
\end{split}
\end{equation}
as $z\rightarrow\infty^2$ in the first quadrant. On the other hand,
since $\tilde{N}_j(z)$ in (\ref{eq:tildeN}) is a constant
$\mathcal{K}_j$, we also have
\begin{equation}\label{eq:Nhatinf2}
\begin{split}
\hat{N}_2(z)&=\mathcal{K}_2L_2^{-1}\left(z^{\frac{1}{3}}+\kappa_{2,0}-\kappa
z^{-\frac{1}{3}}+O(w^2)\right),\\
\hat{N}_3(z)&=\mathcal{K}_3L_3^{-1}\left(1+\kappa_{3,0}z^{-\frac{1}{3}}+O(w^2)\right),
\end{split}
\end{equation}
as $z\rightarrow\infty^2$ in the first quadrant. By comparing
(\ref{eq:Nhatinf}) and (\ref{eq:Nhatinf2}), we obtain
(\ref{eq:Nbeha}). This concludes the proof of the theorem.
\end{proof}
\section{The non-vanishing of the theta function}\label{se:nonvan}
We will now prove that the normalization constants
$\theta\left(u(\infty^k)+\frac{\beta_j}{2\pi
i}+n\vec{\alpha}\right)$, $j=1,\ldots,4$ and $k=1,2$ does not vanish
for any $n\in\mathbb{N}$. Then the solution $S^{\infty}(z)$ of the
Riemann-Hilbert problem (\ref{eq:outer}) constructed in Theorem
\ref{thm:outer} exists and is well-defined. We will then show that
it satisfies the conditions in Theorem \ref{thm:DK}. We will use the
results in Chapter 6 of \cite{Fay}.

First let us define a contour $\Gamma$ that divides the Riemann
surface $\mathcal{\Lie}$ into 2 halves. Let $\Gamma$ be the set of
points that is fixed under the map $\phi$ in (\ref{eq:phi}). That
is,
\begin{equation}\label{eq:gammadef}
\Gamma=\left\{x\in\mathcal{\Lie}|\quad \phi(x)=x\right\}
\end{equation}
Then $\Gamma$ is a disjoint union of $g+1$ closed curves $\Gamma_j$,
$j=0,\ldots,g$ on $\mathcal{\Lie}$, given by the followings.
\begin{equation}\label{eq:Gamma}
\begin{split}
\Gamma&=\cup_{j=0}^{g}\Gamma_j,\\
\Gamma_0&=\cup_{k=1}^2\left(\tilde{\Xi}_0^k\cup\tilde{\Xi}_{g+1}^k\right),\\
\Gamma_j&=\cup_{k=1}^2\tilde{\Xi}_j^k,\quad j=1,\ldots, g.
\end{split}
\end{equation}
where $\tilde{\Xi}_j^k$ is the contour on $\mathcal{\Lie}_k$ that
projects to $\tilde{\Xi}_j$. That is
\begin{equation*}
\begin{split}
\tilde{\Xi}_j^k=\left\{(z,\xi)|z\in\tilde{\Xi}_j,\quad
\xi=\xi_k(z)\right\}.
\end{split}
\end{equation*}
In other words, the contours $\Gamma_j$ are the closed loops on
$\mathcal{\Lie}$ that start from the branch point $\lambda_{2j}^1$,
going through the interval $[\lambda_{2j},\lambda_{2j+1}]$ on
$\mathcal{\Lie}_1$, then enters $\mathcal{\Lie}_2$ at
$\lambda_{2j+1}^1$ and go back to $\lambda_{2j}^1$ through the
interval $[\lambda_{2j},\lambda_{2j+1}]$ on $\mathcal{\Lie}_2$. The
contour $\Gamma_0$ starts at $\lambda_1^1$, goes to $-\infty$ on the
real axis on $\mathcal{\Lie}_1$, then from $+\infty$ to
$\lambda_{2g+2}^1$ on the real axis on $\mathcal{\Lie}_1$, from
which it enters $\mathcal{\Lie}_2$ and goes to $+\infty$ along the
real axis on $\mathcal{\Lie}_2$, then goes back from $-\infty$ on
$\mathcal{\Lie}_2$ to $\lambda_{1}^1$ along the real axis. (See
Figure \ref{fig:Gamma}).

\begin{figure}
\centering \psfrag{lambda1}[][][1][0.0]{\tiny$\lambda_1$}
\psfrag{l2}[][][1][0.0]{\tiny$\lambda_2$}
\psfrag{l3}[][][1][0.0]{\tiny$\lambda_3$}
\psfrag{l4}[][][1][0.0]{\tiny$\lambda_4$}
\psfrag{l2g}[][][1][0.0]{\tiny$\lambda_{2g+2}$}
\psfrag{g0}[][][1][0.0]{\tiny$\Gamma_0$}
\psfrag{g1}[][][1][0.0]{\tiny$\Gamma_1$}
\psfrag{g2}[][][1][0.0]{\tiny$\Gamma_2$}
\psfrag{g3}[][][1][0.0]{\tiny$\Gamma_{g}$}
\psfrag{it}[][][1][0.0]{\tiny$ic$}
\psfrag{-it}[][][1][0.0]{\tiny$-ic$}
\includegraphics[scale=1]{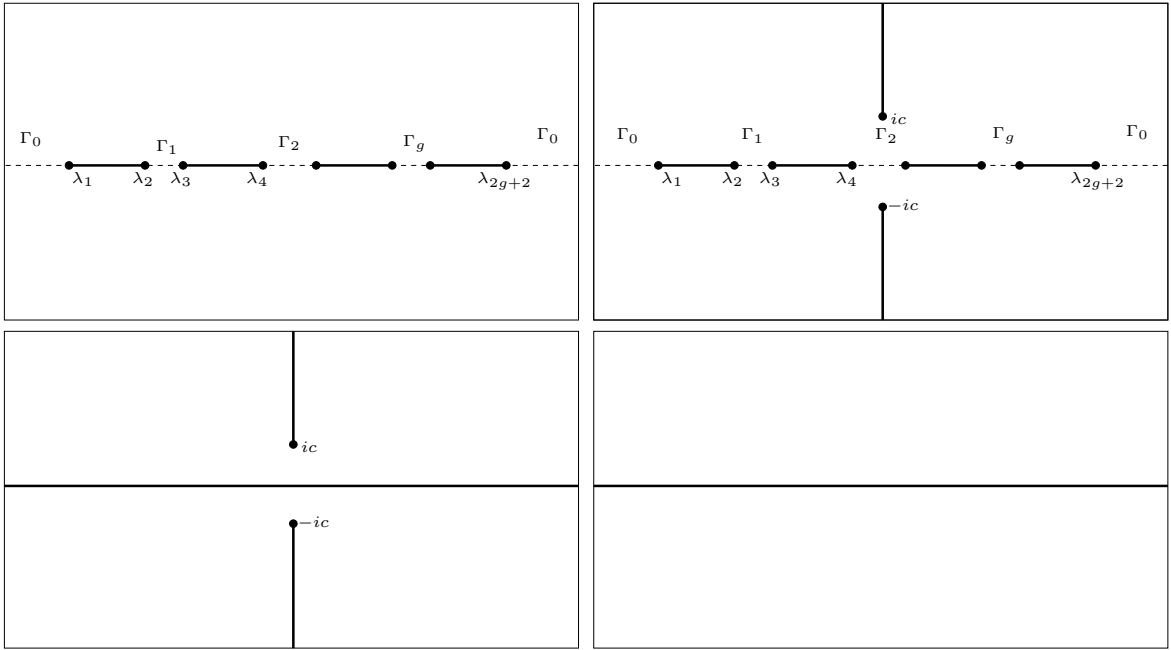}
\caption{The dash lines indicate the loops
$\Gamma_j$.}\label{fig:Gamma}
\end{figure}

Note that the images of the cuts $\Xi_j$ on $\mathcal{\Lie}_1$ and
$\mathcal{\Lie}_2$ do not belong to $\Gamma$. For example, let
$x=(z,\xi_{1,+}(z))$ be a point on $\Xi_j^1$, then
\begin{equation*}
\phi(x)=(\overline{z},\xi_{1,+}(\overline{z}))=(z,\xi_{1,-}(z))\neq
x.
\end{equation*}
Similarly, the images of the real axis on $\mathcal{\Lie}_3$ and
$\mathcal{\Lie}_4$ do not belong to $\Gamma$ either.

The curve $\Gamma$ divides the Riemann surface $\mathcal{\Lie}$ into
2 halves, $\mathcal{\Lie}_+$ and $\mathcal{\Lie}_-$, each of which
is an open Riemann surface with boundary $\Gamma$. The Riemann
surface $\mathcal{\Lie}_{\pm}$ consists of the upper (lower) half
planes of $\mathcal{\Lie}_1$, $\mathcal{\Lie}_2$ and
$\mathcal{\Lie}_3$ and the lower (upper) half plane of
$\mathcal{\Lie}_4$. The Riemann surface $\mathcal{\Lie}$ can now be
thought of as a union of $\mathcal{\Lie}_+$, $\mathcal{\Lie}_-$ and
$\Gamma$. Moreover, the $a$-cycles defined in Figure \ref{fig:cycle}
is homologic to the contours $\Gamma_j$. That is, we have
\begin{equation}\label{eq:gammaa}
\Gamma_j\sim a_j,\quad j=1,\ldots, g.
\end{equation}
We can think of $\mathcal{\Lie}$ as the Riemann surface formed by
gluing two copies of $\mathcal{\Lie}_+$ along the boundary $\Gamma$
with an anti-holomorphic involution $\phi$ that fixes $\Gamma$ and
maps $\mathcal{\Lie}_+$ onto $\mathcal{\Lie}_-$. A Riemann surface
formed in this way is called a Schottky double. Since
$\mathcal{\Lie}$ is a Schottky double, we can apply the results in
Chapter 6 of \cite{Fay} to the theta function of $\mathcal{\Lie}$.

Let us define the tori $\mathbb{S}_{\chi}$ and $\mathbb{T}_{\chi}$
as in Propositions 6.2 and 6.8 of \cite{Fay}.
\begin{definition}\label{de:tori}
Let
$\chi=(\chi_1,\ldots,\chi_{g})^T\in\left(\mathbb{Z}/2\mathbb{Z}\right)^g$
and let $\mathbb{J}_0$ be the torus
\begin{equation}
\mathbb{J}_0=\mathbb{C}^g/\Lambda,\quad
\Lambda=\mathbb{Z}^g+\mathbb{Z}^g\Pi.
\end{equation}
The tori $\mathbb{S}_{\chi}$ and $\mathbb{T}_{\chi}$ are tori in
$\mathbb{J}_0$ defined by
\begin{equation}\label{eq:ST}
\begin{split}
\mathbb{S}_{\chi}&=\left\{\vec{s}\in\mathbb{J}_0,|\quad
\vec{s}=\frac{1}{2}\chi+\Pi\varsigma,\quad
\varsigma\in\mathbb{R}^g.\right\},\\
\mathbb{T}_{\chi}&=\left\{\vec{t}\in\mathbb{J}_0,|\quad
\vec{t}=\varsigma+\frac{1}{2}\Pi\chi,\quad
\varsigma\in\mathbb{R}^g.\right\}.
\end{split}
\end{equation}
\end{definition}
Note that this definition is different from the one in \cite{Fay}
because the theta function in \cite{Fay} is defined differently.

We can now apply the results in \cite{Fay}. The first result tells
us where the zeros $\iota_j$ of the function $\theta(u(x))$ are
located.
\begin{proposition}\label{pro:prefay2}
(Proposition 6.4 of $\cite{Fay}$) For any point $x_0\in\Gamma_0$,
$\vec{s}\in S_{\chi}$, the function $\theta(u(x)-u(x_0)-\vec{s})$
either vanishes identically or has modulo 2, $1+\chi_k$ zeros on
$\Gamma_k$, where $\chi_k$ is the $k^{th}$ component of the vector
$\chi$.
\end{proposition}
As a corollary, we have the following concerning the locations of
the zeros $\iota_j$.
\begin{corollary}\label{pro:fay2}
The function $\theta(u(x))$ has $g$ zeros $\iota_1,\ldots,\iota_g$
such that $\iota_k\in\Gamma_k$, $k=1,\ldots,g$.
\end{corollary}
\begin{proof} Let us take
$x_0=\lambda_{2g+2}^1$, $\chi=0$ and $\vec{s}=0$ in Proposition
\ref{pro:prefay2}, then $u(x_0)=0$ and by the paragraph after Lemma
\ref{le:imag}, we see that $\theta(u(x))$ is not identically zero
and hence by Proposition \ref{pro:prefay2}, it has $1$ zero on each
of the contour $\Gamma_k$, $k=1,\ldots, g$.
\end{proof}
The next result shows that the theta function does not vanish when
its argument is real.
\begin{proposition}\label{pro:fay1}
(Corollary 6.13 of $\cite{Fay}$) Let $\hat{\mathbb{T}}_0$ be the
universal covering of $\mathbb{T}_0$,
\begin{equation}
\hat{\mathbb{T}}_{0}=\left\{\vec{t}\in\mathbb{C}^g,|\quad
\vec{t}=\varsigma,\quad \varsigma\in\mathbb{R}^g.\right\}.
\end{equation}
Then the theta function $\theta(\vec{t})$ is real and positive for
$\vec{t}\in\hat{\mathbb{T}}_0$. That is, $\theta(\vec{t})$ is real
and positive for all $\vec{t}\in\mathbb{R}^g$.
\end{proposition}
We can now prove that the periods $\beta_j$ of $\D\Delta_j$ in
Definition \ref{de:deltaj} are purely imaginary. This, together with
Proposition \ref{pro:fay1} will imply the non-vanishing of the theta
functions $\theta\left(u(\infty^k)+\frac{\beta_j}{2\pi
i}+n\vec{\alpha}\right)$ for $k=1,2$ and $j=1,\ldots,4$.
\begin{lemma}\label{pro:purim}
The periods $\beta_j$ of the 1-forms $\D\Delta_j$ defined in
Definition \ref{de:deltaj} are purely imaginary.
\end{lemma}
\begin{proof} First note that, by Corollary \ref{pro:fay2}, all
the points $\iota_l$ and $\lambda_l^1$ are invariant under the
involution $\phi$. Hence the meromorphic 1-form
$\D\tilde{\Delta}_j=\overline{\D\Delta_j(\phi(x))}$ has the same
poles and residues as $\D\Delta_j(x)$. Let us show that the
$a$-periods of $\D\tilde{\Delta}_j$ are zero. We have
\begin{equation}\label{eq:deltaper}
\oint_{a_k}\overline{\D\Delta_j(\phi(x))}=\oint_{\phi(a_k)}\overline{\D\Delta_j(x)}
\end{equation}
\begin{figure}
\centering \psfrag{lambda1}[][][1][0.0]{\tiny$\lambda_1$}
\psfrag{l2}[][][1][0.0]{\tiny$\lambda_2$}
\psfrag{l3}[][][1][0.0]{\tiny$\lambda_3$}
\psfrag{l4}[][][1][0.0]{\tiny$\lambda_4$}
\psfrag{l2g}[][][1][0.0]{\tiny$\lambda_{2g+2}$}
\psfrag{i1}[][][1][0.0]{\tiny$\iota_1$}
\psfrag{it}[][][1][0.0]{\tiny$ic$}
\psfrag{-it}[][][1][0.0]{\tiny$-ic$}
\includegraphics[scale=0.75]{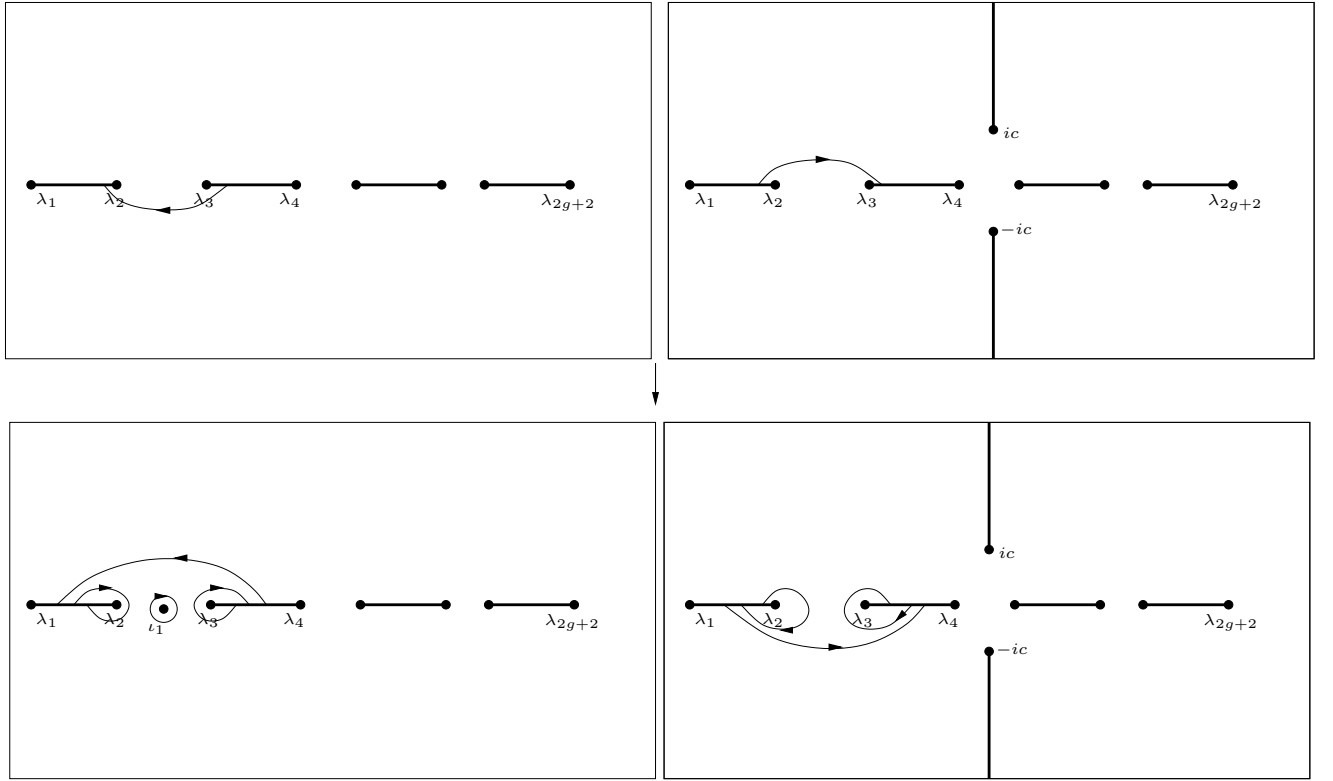}
\caption{The deformation from the cycle $\phi(a_1)$ to $a_1$ when
$\iota_1\in\tilde{\Xi}_1^1$. The deformations for other cycles are
similar.}\label{fig:phia}
\end{figure}
From Figure \ref{fig:cycle}, we see that the curve $\phi(a_k)$
consists of a path from the lower half plane in $\mathcal{\Lie}_1$
that goes from $\Xi_{k+1}$ to $\Xi_{k}$, and another path in the
upper half plane of $\mathcal{\Lie}_2$ that goes from $\Xi_{k}$ to
$\Xi_{k+1}$. There are 3 poles of $\D\Delta_j$ between the loops
$a_k$ and $\phi(a_k)$: $\iota_k$, $\lambda_{2k}^1$ and
$\lambda_{2k+1}^1$ (See Figure \ref{fig:phia}). From
(\ref{eq:deltaj}), we see that $\D\Delta_j$ has a combined residue
of 0 at these points, and hence we can deform $\phi(a_k)$ onto $a_k$
without affecting the value of (\ref{eq:deltaper}). Therefore, by
(\ref{eq:deltaper}), we see that
\begin{equation*}
\oint_{a_k}\overline{\D\Delta_j(\phi(x))}=0,\quad j=1,\ldots,g.
\end{equation*}
By the uniqueness of normalized 1-form, this implies
$\overline{\D\Delta_j(\phi(x))}=\D\Delta_j$. Now we use
(\ref{eq:involcyc}) for the $b$-periods, since the relations for the
$b$-cycles in (\ref{eq:involcyc}) are exact and not up to
deformation, we have
\begin{equation*}
\oint_{b_k}\overline{\D\Delta_j(\phi(x))}=-\oint_{b_k}\overline{\D\Delta_j(x)}=-\overline{\left(\beta_j\right)_k}.
\end{equation*}
where $\left(\beta_j\right)_k$ is the $k^{th}$ component of the
vector $\beta_j$. On the other hand, since
$\overline{\D\Delta_j(\phi(x))}=\D\Delta_j$, the above is also equal
to $\left(\beta_j\right)_k$. This implies the proposition.
\end{proof}
From Proposition \ref{pro:fay1} and Lemma \ref{pro:purim}, we obtain
\begin{theorem}\label{thm:nonvan}
There exists $\delta>0$, independent on $n$, such that
$\theta\left(u(\infty^k)+\frac{\beta_j}{2\pi
i}+n\vec{\alpha}\right)>\delta$, for $k=1,2$ and $j=1,\ldots,4$ and
all $n\in\mathbb{N}$.
\end{theorem}
\begin{proof} Let us consider the normalized 1-form $\D\Omega_k$
that has simple poles at $\lambda_{2g+2}^1$ and $\infty^k$ with
residues -1 and 1 for $k=1,2$. Then by similar argument used in the
proof of Lemma \ref{pro:purim}, we see that
$\overline{\D\Omega_k(\phi(x))}=\D\Omega_k$ and hence the
$b$-periods of $\D\Omega_k$ are all purely imaginary. Now by the
Riemann bilinear formula (\ref{eq:RB}) and the definition of the
Abel map (\ref{eq:abel}), we see that
\begin{equation*}
2\pi i\left(u_l(\infty^k)-u_l(\lambda_{2g+2}^1)\right)=2\pi
iu_l(\infty^k)=\oint_{b_l}\D\Omega_k,\quad k=1,2,\quad l=1,\ldots,g.
\end{equation*}
Hence $u(\infty^k)$, $k=1,2$ are real.

Therefore, by Proposition \ref{pro:fay1} and Lemma \ref{pro:purim},
we see that $\theta\left(u(\infty^k)+\frac{\beta_j}{2\pi
i}+n\vec{\alpha}\right)>0$. By the periodicity of the theta function
(\ref{eq:period}), we see that the theta function is in fact a map
from $T\times \mathbb{R}^g\rightarrow\mathbb{C}$, where $T$ is the
torus $T=\mathbb{R}^g/\mathbb{Z}^g$. By Proposition \ref{pro:fay1},
the restriction of the theta function on the compact set
$T\times\{0,0,\ldots,0\}$ is real and positive and hence there
exists $\delta>0$ such that $\theta(\vec{t})>\delta$ for all
$\vec{t}\in T$. This then implies the theorem.
\end{proof}
This implies that the function $S^{\infty}(z)$ in Theorem
\ref{thm:outer} exists. We will now show that it satisfies the
conditions in Theorem \ref{thm:DK}.
\begin{corollary}\label{cor:bound}
The function $S^{\infty}(z)$ in (\ref{eq:sol}) and its inverse
$\left(S^{\infty}(z)\right)^{-1}$ satisfy the conditions in Theorem
\ref{thm:DK}.
\end{corollary}
\begin{proof}
Let us first show that the function $N(z)$ in (\ref{eq:nmatrix}) is
bounded in $n$ uniformly in $\mathcal{T}$, where $\mathcal{T}$ is
defined in (\ref{eq:calT}). Since the entries of $N(z)$ are
restrictions of the functions $N_j(z)$ in (\ref{eq:Nj}) on different
sheets of the Riemann surface, we only need to show that $N_j(z)$ is
bounded inside the set
$\hat{\mathcal{T}}=\cup_{l=1}^4\xi_l(\mathcal{T})$ in
$\mathcal{\Lie}$ that projects onto $\mathcal{T}$. From the
periodicity property of the theta function (\ref{eq:period}), we see
that $N_j(z)$ can be written as
\begin{equation}\label{eq:Nmod}
\begin{split}
N_j(z)&=e^{\Delta_j(z)}\frac{\theta\left(u(z)+\frac{\beta_j}{2\pi
i}+n\vec{\alpha}\right)}
{\theta\left(u(z)\right)}\\
&=e^{\Delta_j(z)}\frac{\theta\left(u(z)+\frac{\beta_j}{2\pi
i}+\vec{\gamma}_n\right)} {\theta\left(u(z)\right)},\quad
j=1,\ldots,4.
\end{split}
\end{equation}
where $\vec{\gamma}_n$ is a finite vector given by
\begin{equation*}
\left(\vec{\gamma}_n\right)_l=n\alpha_l-[n\alpha_l],\quad
l=1,\ldots,g,
\end{equation*}
where $[x]$ is the biggest integer that is smaller than $x$. From
(\ref{eq:Nmod}) and the fact that $\theta(u(x))$ is not identically
zero (Proposition \ref{pro:constant}), we see that $N_j(z)$ is
bounded in $n$ uniformly in $\hat{\mathcal{T}}$.

We will now show that the constants $L_1,\ldots,L_4$ in
(\ref{eq:Lexp}) are bounded in $n$. From the singularity behavior of
the meromorphic 1-forms $\D\Delta_j$ in (\ref{eq:deltaj}), we see
that following constants
\begin{equation*}
e^{-\Delta_1(\infty^1)},\quad \lim_{w\rightarrow
0}e^{-\Delta_2(w)w^{-1}},\quad \lim_{w\rightarrow
0}e^{-\Delta_3(w)},\quad \lim_{w\rightarrow 0}e^{-\Delta_4(w)w},
\end{equation*}
in (\ref{eq:Lexp}) are all bounded and non-zero. Since they are all
independent on $n$, they are also bounded away from infinity and
zero as $n\rightarrow\infty$. By Proposition \ref{pro:constant}, we
see that $\theta(u(x))$ is not identically zero and will only vanish
at the points $\iota_l$ that belong to $\Gamma_l$. Since neither
$\infty^1$ nor $\infty^2$ belongs to $\Gamma_l$ for $l=1,\ldots, g$,
the constants $\theta(u(\infty^k))$, $k=1,2$ are non-zero. Moreover,
from the definition of the Abel map (\ref{eq:abel}), we see that
$u(\infty^1)$ and $u(\infty^2)$ are both finite and hence
$\theta(u(\infty^1))$ and $\theta(u(\infty^2))$ are both bounded and
are independent on $n$. Let us now consider the factors
$\theta\left(u(\infty^k)+\frac{\beta_j}{2\pi
i}+n\vec{\alpha}\right)$ for $k=1,2$ and $j=1,\ldots,4$. By Theorem
\ref{thm:nonvan}, there exists $\delta>0$, independent on $n$ such
that these constants are greater than $\delta$. On the other hand,
from the periodicity of the theta function (\ref{eq:period}) and the
fact that the period matrix $\Pi$ is purely imaginary, (Lemma
\ref{le:imag}) while the vector $\alpha$ in (\ref{eq:apj}) is real,
we see that $\theta\left(u(\infty^k)+\frac{\beta_j}{2\pi
i}+n\vec{\alpha}\right)$ is bounded in $n$ as $n\rightarrow\infty$.
Hence the constants $L_1,\ldots,L_4$ in (\ref{eq:sol}) and
(\ref{eq:Lexp}) are bounded away from infinity and zero as
$n\rightarrow\infty$.

Finally, by considering the asymptotic expansion of $N_j(z)$ in the
local parameter $w$ in (\ref{eq:w}) at $z=\infty$ and making use of
(\ref{eq:Nmod}), we see that $\kappa$ in (\ref{eq:sol}) is bounded
in $n$ as $n\rightarrow\infty$. Since all the constants $L_j$ and
$\kappa$ are bounded in $n$ as $n\rightarrow\infty$ and that all the
$N_j(z)$ are bounded in $n$ uniformly in $\hat{\mathcal{T}}$, we see
that $S^{\infty}(z)$ is also bounded in $n$ uniformly in
$\mathcal{T}$. To see that this is also the case for the inverse
$\left(S^{\infty}(z)\right)^{-1}$, let us consider the determinant
of $S^{\infty}(z)$. Since $S^{\infty}(z)$ is a solution to the
Riemann-Hilbert problem (\ref{eq:outer}), the determinant
$\det\left(S^{\infty}(z)\right)$ has no jump discontinuity in
$\mathbb{C}$ and it behaves as $1+O(z^{-\frac{1}{3}})$ as
$z\rightarrow\infty$. From the expression of $N(z)$ in
(\ref{eq:nmatrix}), we see that at $\lambda_{j}$, only the first and
second columns of $N(z)$ have fourth-root singularities, while at
the points $\pm ic$, only the second and the third columns of $N(z)$
have fourth-root singularities. Therefore the determinant of
$S^{\infty}(z)$ can at worst have square-root singularities at these
points. Since $\det\left(S^{\infty}(z)\right)$ has no jump
discontinuities in $\mathbb{C}$, we see that
$\det\left(S^{\infty}(z)\right)$ cannot have square-root
singularities at these points. Hence
$\det\left(S^{\infty}(z)\right)$ is holomorphic in the whole complex
plane. By Liouville's theorem, this implies that
$\det\left(S^{\infty}(z)\right)=1$. Since the entries of
$\left(S^{\infty}(z)\right)^{-1}$ are degree 3 polynomials in the
entries of $S^{\infty}(z)$ divided by
$\det\left(S^{\infty}(z)\right)=1$, we see that the entries of
$\left(S^{\infty}(z)\right)^{-1}$ are also bounded in $n$ uniformly
in $\mathcal{T}$.

Finally, by considering the asymptotic expansion of $N_j(z)$ in the
local parameter $w$ in (\ref{eq:w}) at $z=\infty$ and making use of
(\ref{eq:Nmod}), it is easy to see that condition 2. in Theorem
\ref{thm:DK} is satisfied for $S^{\infty}(z)$ and its inverse.
\end{proof}
We can now use Theorem \ref{thm:DK} to conclude that Theorem
\ref{thm:main1} and Theorem \ref{thm:main2} are true.

\end{document}